\documentclass[10pt]{article}
\usepackage[utf8]{inputenc}

\usepackage{conf-abbrvs}
\usepackage{custom-macro-prev}
\title{{\bf Log Diameter Rounds MST Verification and Sensitivity in MPC}}
\date{}
\author{
\begin{tabular}[t]{c c c c}
\textbf{Sam Coy}\thanks{Department of Computer Science, University of Warwick, UK. E-mail: S.Coy@warwick.ac.uk. Research supported in part by the Centre for Discrete Mathematics and its Applications (DIMAP) and 
by an EPSRC studentship.} & 
\textbf{Artur Czumaj}\thanks{Department of Computer Science and Centre for Discrete Mathematics and its Applications (DIMAP), University of Warwick, UK. E-mail: A.Czumaj@warwick.ac.uk. Research supported in part by the Centre for Discrete Mathematics and its Applications (DIMAP), by EPSRC award EP/V01305X/1, by a Weizmann-UK Making Connections Grant, by an IBM Award.} &
\textbf{Gopinath Mishra}\thanks{Department of Computer Science, National University of Singapore, Singapore. E-mail:  Gopinath@nus.edu.sg. The work was partly done when the author was a Postdoc at the University of Warwick and the research was supported in part by the Centre for Discrete Mathematics and its Applications (DIMAP) and by EPSRC award EP/V01305X/1.}  &
\textbf{Anish Mukherjee}\thanks{Department of Computer Science, University of Warwick, UK. E-mail: Anish.Mukherjee@warwick.ac.uk. Research supported in part by the Centre for Discrete Mathematics and its Applications (DIMAP), by EPSRC award EP/V01305X/1.} \\
{\small University of } & {\small University of } & {\small National University of } &   {\small University of } \\
{\small Warwick} & {\small Warwick}& {\small Singapore} & {\small Warwick}
\end{tabular}
}

\begin{document}
\maketitle

\begin{abstract}
We consider two natural variants of the problem of minimum spanning tree (\MST) of a graph in the parallel setting: \emph{\MST verification} (verifying if {a given tree is an \MST}) and the \emph{sensitivity analysis of an \MST} (finding the lowest cost replacement edge for each edge of the \MST). These two problems have been studied extensively for sequential algorithms and for parallel algorithms in the \PRAM model of computation. In this paper, we extend the study to the standard model of \emph{Massive Parallel Computation} (\MPC).

It is known that for graphs of diameter $D$, the connectivity problem can be solved in $O(\log D + \log\log n)$ rounds on an \MPC with \emph{low local memory} (each machine can store only $O(n^{\delta})$ words for an arbitrary constant $\delta > 0$) and with \emph{linear global memory}, that is, with \emph{optimal utilization}. {However, for the related task of finding an \MST,  we need $\Omega(\log D_{\MST})$ rounds, where $D_{\MST}$ denotes the diameter of the minimum spanning tree}. The state of the art upper bound for \MST is $O(\log n)$ rounds; the result follows by simulating existing \PRAM algorithms. While this bound may be optimal for general graphs, {the benchmark of connectivity and lower bound for \MST suggest the target bound of $O(\log D_\MST)$ rounds, or possibly $O(\log D_\MST + \log\log n)$ rounds}. As for now, we do not know if this bound is achievable for the \MST problem on an \MPC with low local memory and linear global memory. In this paper, we show that two natural variants of the \MST problem: \MST verification and sensitivity analysis of an \MST, can be completed in $O(\log D_T)$ rounds on an \MPC with low local memory and with linear global memory, that is, with optimal utilization; here $D_T$ is the diameter of the input ``candidate \MST'' $T$. The algorithms asymptotically match our lower bound, conditioned on the 1-vs-2-cycle conjecture.
\end{abstract}

\section{Introduction}
Over the last decade, the heightened need to solve problems on ever-increasing volumes of data has motivated the successful development of massively parallel computation frameworks, such as MapReduce, Hadoop, Dryad, and Spark, see, e.g., \cite{DG08, White12, ZCFSS10}.
In turn, the prevalence of these frameworks has led to a huge interest in the design of efficient parallel algorithms.

The \emph{Massively Parallel Computing (\MPC)} model is a theoretical model which abstracts these massively parallel computation frameworks.
The model was first introduced by Karloff, Suri, and Vassilvitskii \cite{KSV10}, and is now the standard theoretical model for these systems (see also, e.g., \cite{ASSWZ18,GSZ11,dp_trees_mpc}).
An \MPC is a parallel computing system with \machines \emph{machines}.
Each of these machines has \lspace words of \emph{local memory}; in total the system has $\machines \cdot \lspace = \gspace$ \emph{global memory}.
The input to an \MPC algorithm is arbitrarily distributed among the machines.
Computation in the \MPC model takes place in synchronous \emph{rounds}.
In each round, each machine does some computation on the data in its local memory.
Then, machines exchange messages: each machine may send as many $1$-word messages as it likes to any other machine, as long as in total, each machine sends and receives at most \lspace words in any given round.
The complexity of an \MPC algorithm is measured in the number of rounds.

For graph problems, our input is an edge-weighted graph $G = (V, E)$ with $n$ vertices, $m$ edges, and weighting function $w$: we assume that the weight of any edge can be expressed with $O(1)$ words. In this paper, we focus on the design of \emph{low-space} (fully scalable) \MPC algorithms, where the local memory of each machine is $\lspace = O(n^\delta)$ for some (arbitrary) constant $\delta \in (0, 1)$. Our focus is on designing \emph{algorithms with optimal utilization}, that is, those achieving linear global memory: for graph problems, this means that $\gspace = \Theta(m+n)$. The goal is to design algorithms that for a given \MPC with for local memory $\lspace = O(n^\delta)$ and with linear global memory \gspace minimizes its complexity (the number of rounds performed). Of course, optimally, we would like to achieve a constant number of rounds, but for many problems, the best we could hope for is a sub-logarithmic number of rounds.

The low-space \MPC regime, which we consider in this paper, is the most challenging of the local memory regimes usually considered in the \MPC model.
Nevertheless, it is a more powerful model than the \PRAM model: a step of \PRAM can be simulated in $O(1)$ rounds of the \MPC model 
and sorting or computing the prefix-sum of $N$ values can be done in $O(1)$ rounds on an \MPC with local memory $\lspace = O(N^\delta)$, see, e.g., \cite{GSZ11}.

In this paper, we study the complexity of important variants of the fundamental graph problem of \emph{Minimum Spanning Tree (\MST)}. The \MST problem has been extensively studied for many decades. After many years of extensive research, we know the problem can be solved by a randomized algorithm in $O(m)$-time \cite{KKT95}, though we still do not know the exact deterministic complexity of the \MST problem (see, e.g., \cite{Chazelle00,PR02}). Similar efforts have been also taking place for parallel algorithms, largely focusing on \PRAM algorithms, where $O(\log n)$-time algorithms can be achieved with almost optimal total work. Given the importance of the \MST problem, there has been also extensive research studying its variations, most prominently the natural variants of \emph{\MST verification} and \emph{\MST sensitivity} problems, see, e.g., \cite{DT97,DRT92,King97,KPRS97,Komlos84,Pettie15}.

\begin{definition}[\textbf{\MST Verification}]
\label{def:mst-verification}
Given a graph $G=(V, E)$ with edges weighted by a function $w: E \rightarrow \mathbb{N}$, and a tree $T \subseteq E$, the goal is to determine whether $T$ is an \MST of $G$.
\end{definition}

\begin{definition}[\textbf{\MST Sensitivity}]
\label{def:mst_sensitivity}
The \MST sensitivity problem is for a given graph $G=(V, E)$ and an \MST $T \subseteq E$, to calculate a \emph{sensitivity} value $\sens(e)$ for all $e \in E$. We define $\sens(e)$ as:
\begin{itemize}
\item for $e \in T$: the amount by which $w(e)$ would have to be increased (with all other edges remaining the same weight) for $e$ to no longer be part of an \MST of $G$;
\item for $e \in E \setminus T$: the amount by which $w(e)$ would have to be decreased (with all other edges remaining the same weight) for $e$ to be part of an \MST of $G$.
\end{itemize}
\end{definition}

The \MST problem has been studied also in the \MPC model, extending the earlier research for the sequential model and the \PRAM.
For super-linear local space ($\lspace = O(n^{1+\delta})$) and linear local space ($\lspace = O(n)$), 
constant-round algorithms have been developed \cite{LMSV11,JN18,Nowicki21} with optimal utilization. However, for the setting of low-space \MPC, which is the main focus of our paper, the situation is fundamentally different. Let $\dmst$ denote the smallest diameter of a minimum spanning tree of $G$. We do not expect a constant-rounds \MPC algorithm to exist and while there is an $O(\log \dmst)$ round algorithm (optimal according to a conditional lower bound) if we are allowed polynomial extra global space \cite{ASSWZ18,CoyC23,behnezhad_connectivity}, the best result known with optimal global memory and utilization runs in $O(\log n)$ rounds (e.g., by a simulation of \PRAM algorithms). This is quite a large gap, especially for low-diameter network topologies which are common in real-world settings, and whether this gap can be closed is a major open problem in the area of \MPC algorithms. In this paper, similarly as it has been done before for sequential and \PRAM models, we study two natural variants of the minimum spanning tree (\MST) problem in the \MPC model: the \MST verification and \MST sensitivity problems.

\subsection{Our Results}

We demonstrate that while we do not know whether the \MST problem can be solved in an optimal number of rounds on an \MPC in the low-memory regime and with optimal utilization, the \MST verification and \MST sensitivity problems can be solved optimally.\footnote{Throughout this paper, by \emph{high probability} we mean the result is correct with probability at least $1-1/\poly(n)$.}

Our first main result is an algorithm for the \MST verification.

\showresult{\Cref{thm:mst_verification}}{\textbf{\MST Verification Algorithm}}{
Let $\delta$ be an arbitrarily small positive constant. Let $G = (V, E)$ be an edge-weighted input graph and let $T \subseteq E$ be a tree. Suppose the diameter of $T$ is $\dt$.

One can decide with high probability whether $T$ is an \MST of $G$ in $O(\log \dt)$ rounds on an \MPC with local memory $\lspace = O(n^\delta)$ and optimal global memory $\gspace = O(m+n)$.
}

We extend our verification algorithm to the more challenging problem of \MST sensitivity, still obtaining a $O(\log \dt)$ round algorithm that uses optimal global memory:

\showresult{\Cref{thm:mst_sensitivity}}{\textbf{\MST Sensitivity Algorithm}}{
Let $\delta$ be an arbitrarily small positive constant. Let $G = (V, E)$ be an edge-weighted input graph and let $T \subseteq E$ be an \MST of $G$ with diameter $\dt$.

The problem of \MST sensitivity can be solved in $O(\log \dt)$ rounds (with high probability) on an \MPC with local memory $\lspace = O(n^\delta)$ and optimal global memory $\gspace = O(m+n)$.
}

Finally, we complement \Cref{thm:mst_verification} by showing that this result is optimal: conditioned on the widely believed 1-vs-2-cycle conjecture, $\Omega(\log \dt)$ rounds are required for \MST verification if our input $T$ is an arbitrary set of $n-1$ edges of $G$.

\showresult{\Cref{thm:verification_lb}}{\textbf{\MST Verification Lower Bound}}{
Let $\delta < 1$ be an arbitrary positive constant. Let $G = (V, E)$ be an edge-weighted graph of diameter $D$.
Let $T \subseteq E$ be a set of $n-1$ edges of $G$ and let $\dt$ be the diameter of $T$.

Deciding whether $T$ is a minimum spanning forest of $G$ requires $\Omega(\log \dt)$ rounds, on an \MPC with local memory $\lspace = O(n^\delta)$, unless the 1-vs-2-cycle conjecture is false.
}

We remark that our algorithms can also deal with the problems above even when the input graph $G$ is not connected, in which case our goal is to study the variants of the \emph{minimum spanning forest (MSF)} problem. All our algorithms can be easily extended to this more general setting since it is known that finding the connected components in a forest can be solved optimally in the sublinear-memory \MPC model; see also \Cref{rem:extend_to_forests}.

%
%
\subsection{Technical Challenges and Contributions}
To date, the problem of finding an \MST in low-space \MPC has been resistant to efficient algorithms with optimal global memory. Our results on the related problems of \MST verification and sensitivity overcome several interesting technical challenges.
Take \MST verification for example: a naive approach to the problem would involve collecting for each vertex $u$, a path from $u$ to the root of $T$.
Then, we can use prefix-sum in $O(1)$ rounds to, for each ancestor $v$ of $u$, find the largest edge on a path from $u$ to $v$ in $T$.
Then for all non-tree edges $\{u, v\}$, where $v$ is an ancestor of $u$ we can check (using sorting) whether the weight of the non-tree edge is smaller than the largest edge on the path from $u$ to $v$ in $T$.
Doing this for all edges in parallel is the same as checking whether $T$ is an \MST.

The problem with this approach, however, is that it requires $O(\dt)$ memory per vertex, and therefore $O(m + n\cdot \dt)$ global memory; in the worst case, $O(n^2)$ global memory, which is extremely undesirable, particularly for sparse graphs. 
To get around this problem, we use an approach of \emph{hierarchical clustering} (\Cref{sec:clusters}).
Initially, each vertex of the graph is a cluster; for $O(\log \dt)$ iterations, we contract an independent set of edges in the tree to form larger and larger clusters.
Each iteration reduces the number of clusters by a constant factor (by a result of \cite{behnezhad_connectivity,CoyC23}); after $O(\log \dt)$ rounds, the number of clusters has been reduced by a factor of $\poly(\dt)$. 
Then, we have enough global memory per cluster to execute our more naive approach, after which we can (if necessary) ``un-contract'' the clusters and obtain a result for the original graph.

This approach comes with some challenges, however.
In particular, we have to be very careful when performing cluster contraction to preserve enough information to solve the problem.
For example, for \MST verification, for each cluster $c$, we must preserve the highest weight of an edge on all paths from vertices below $c$ to vertices above it.
For \MST sensitivity (\Cref{sec:sensitivity}) and all-edges \LCA (\Cref{sec:all_edges_lca}) it is significantly more challenging to preserve the necessary information.
We note that our hierarchical clustering is different from a recent clustering employed by \cite{dp_trees_mpc} in a paper on dynamic programming in trees (in particular, our clusters are all of height $2$, relative to previous clusters).
We believe that our approach, using clusters to ``compress'' an input tree to allow for greater space-flexibility, could have further applications for efficient algorithms involving trees in low-space \MPC.

\subsection{Related Work}
The \MST problem is a fundamental textbook graph optimization problem that has been extensively studied for many decades (see \cite{GH85} for the historical background of the research before 1985). We know the problem can be solved sequentially by a randomized algorithm in $O(m)$-time \cite{KKT95}, but its deterministic complexity still has some small gap (see though \cite{PR02}), with the fastest algorithm to date \cite{Chazelle00} achieves the running time of $m\alpha(m,n)$, where $\alpha$ is a certain natural inverse of Ackermann's function. This research has been extended to the \PRAM model, where, e.g., Cole \etal \cite{CKT96} gave an optimal $O(\log n)$-time randomized algorithm for the \MST problem.

The problem of verifying if a given spanning tree in an edge-weighted graph is an \MST is very closely related to the \MST problem, but unlike the latter, we understand its sequential and \PRAM complexity very well. Following an early work of Komlos \cite{Komlos84}, Dixon \etal \cite{DRT92} and King \cite{King97} developed deterministic sequential linear-time algorithms for the problem. In parallel setting, Dixon and Tarjan \cite{DT97} gave an optimal $O(\log n)$-time CREW algorithm which was later extended to the EREW \PRAM model \cite{KPRS97}.

For the \MST sensitivity problem, the classical algorithm of Tarjan \cite{Tarjan82} solves it sequentially in $O(m \alpha(m,n))$ time. Dixon \etal \cite{DRT92} designed two \MST sensitivity algorithms, one running in expected linear time and another which is deterministic and provably optimal, but whose complexity is  bounded by $O(m \alpha(m,n))$. Pettie \cite{Pettie15} improved this bound and designed a new \MST sensitivity analysis algorithm running deterministically in $O(m \log\alpha(m,n))$ time.

In the last few years, there has been a series of work extending the study of the \MST problem to the \MPC model.
For the setting of ``super-linear'' \MPC (where each machine has $\lspace = O(n^{1+\delta})$ local memory), Lattanzi \etal \cite{LMSV11} gave a $O(1)$ round algorithm for \MST.

For the setting of linear \MPC, where each machine has $\lspace = O(n)$ memory, a line of work culminated also in $O(1)$ round algorithms for the problem, first randomized \cite{JN18} and then deterministic \cite{Nowicki21}.
Once again, these algorithms relied heavily on the use of connectivity as a subroutine, and these papers also gave randomized and deterministic connectivity algorithms.

For low-space \MPC, which is the focus of this paper, results for \MST and related problems have been limited.
A long and celebrated line of work on connectivity culminated in a $O(\log D + \log \log n)$ round deterministic algorithm for the problem \cite{ASSWZ18, behnezhad_connectivity, CoyC23}, and there is a (conditional) $\Omega(\log D)$ round lower bound \cite{behnezhad_connectivity, CoyC23}. {A recent discovery reveals an optimal $O(\log D)$ rounds algorithm when the input graph is a forest \cite{BLMOU23}.} For \MST, while there is an algorithm which takes $O(\log \dmst)$ rounds \cite{ASSWZ18, CoyC23}, it requires polynomially extra ($\gspace = O(m^{1 + \gamma})$ for any positive constant $\gamma$) global memory, and hence it does not achieve the optimal utilization.
It recursively partitions the input instance into $O(m^{\gamma/2})$ sub-instances of \MST, each of which also has (proportionally) polynomial extra global memory.
At each level of the recursion, computing the connected components of each sub-instance is required: with extra global memory, this takes $O(\log \dmst)$ rounds.
Since the algorithm uses $O(1)$ levels of recursion and each level requires a series of connectivity algorithms with extra global memory to be run in parallel (taking $O(\log \dmst)$ rounds), the algorithm takes $O(\log \dmst)$ rounds in total.
The algorithm can be implemented to work with optimal global memory, but then the number of rounds degrades to $\Omega(\log n)$, since the recursion can branch by at most a constant factor. For global space $\gspace = O(m^{1 + \gamma})$ for $\gamma = o(1)$ and $\gamma = \Omega(\frac{1}{\log n})$, one can find \MST in $O(\frac1{\gamma} \min\{\log\dmst \log(\frac1{\gamma}), \log n\})$ rounds on an \MPC with local memory $\lspace = O(n^\delta)$ (see \cite[Theorem I.7]{ASSWZ18}).
In particular, if one wants to focus on the regime with global memory at most $(n+m) \polylog(n)$ (i.e., $\gamma = O(\frac{\log\log n}{\log n})$), then the number of rounds is $\log n \cdot \min\{\log\dmst, \frac{\log n}{\log\log n}\}$.
Observe that except when the global memory is large, we do not know of any \MPC algorithm (with sublinear local memory) that solves \MST exactly in $o(\log n)$ rounds, even for graphs with small diameter. The approaches used in \cite{ASSWZ18} and other related works \cite{andoni-biconnectivity, BLMOU23, behnezhad_connectivity, CoyC23} that lead to $O(\log D + \log\log n)$-rounds \MPC connectivity algorithms seem unsuitable to obtain even $O(\log n)$-rounds algorithms for \MST. (Still, notice that an optimal $O(\log n)$-rounds \MPC algorithm can be obtained by simulating \PRAM algorithms, e.g., from \cite{CKT96}.)

Similarly to the connectivity conditional lower bound of $\Omega(\log D)$
, there is also a conditional $\Omega(\log \dmst)$ round lower bound for finding an \MST \cite{CoyC23}, which, for the sake of completeness, we sketch in \Cref{app:mst_lower_bound}.

\color{black}
For low-diameter graphs the gap between $O(\log \dmst)$ and $O(\log n)$ is potentially huge, and it is not clear how to even obtain (say) an $O(\dmst)$ round algorithm for the problem. Whether the $\Omega(\log n)$ round barrier can be broken with optimal global memory is a major open question in the area.
\section{Preliminaries and Tools}
In this section, we present some notation and preliminaries, along with two important technical tools: a hierarchical clustering approach, which transforms a tree $T$ into a set of $\frac{n}{\poly \dt}$ clusters; as well as a procedure for finding the lowest common ancestor in the tree of the endpoints of each non-tree edge.

In a rooted tree, let $p(v)$ denote the \emph{parent} of $v$; for the root of the tree $r$, $p(r) = r$. Further, let $p_i(v)$ denote the $i$th ancestor of $v$, i.e. $p_1(v) = p(v)$ and $p_i(v) = p(p_{i-1}(v))$.

We introduce a  notion of a non-tree edge \emph{covering} a tree edge:

\begin{definition}
\label{def:covering}
    We say that a non-tree edge $e$ \emph{covers} a tree edge $t$ if $t$ is one of the edges on the path in $T$ between the endpoints of $e$.
\end{definition}

Observe that the problem of \MST verification is the same as asking: is any tree edge $t$ covered by an edge $e$ with lower weight?

We remark that, without loss of generality, we can assume that our input tree $T$ is a rooted spanning tree:
\begin{remark}
\label{rem:input_is_spanning_tree}
  Our \MST verification  and \MST sensitivity algorithms can assume that $T$ is a rooted spanning tree of $G$.
\end{remark}
\begin{proof}
    We can check whether the given tree $T$ is a spanning tree or not by just counting whether the number of vertices in $T$ is $n$. We can  identify the root of the tree in $O(\log \dt)$ rounds (deterministically, and using optimal global memory) by using tree rooting algorithm of \cite{BLMOU23}.
\end{proof}

Furthermore, in the following, we remark that both algorithms can assume knowledge of $D_T$, the diameter of tree $T$.
\begin{remark}
\label{rem:val-D-T}
    Our \MST verification  and \MST sensitivity algorithms  can  assume the value of $D_T$.
\end{remark}
\begin{proof}
 Recall the connectivity algorithm on forest by \cite{BLMOU23}. The algorithm takes $O(\log D)$ rounds, where $D$ is the diameter of the forest (i.e, the diameter of any component of the forest). Moreover, the algorithm can determine a value $\widehat{D}$ such that $D \leq \widehat{D}\leq 2D$. So, by applying the algorithm of \cite{BLMOU23} on input $T$, we get a 2-approximation to $D_T$ (which is good enough for our \MST verification algorithm  and \MST sensitivity algorithm).    
\end{proof}
 Finally, we remark that all of our algorithms given here can be extended to forests:
\begin{remark}
\label{rem:extend_to_forests}
  Our \MST verification  and \MST sensitivity algorithms can be easily extended to the case where $G$ is disconnected and the input tree $T$ is a forest.
\end{remark}
\begin{proof}
{First, we solve connectivity on forest $T$ by using the algorithm in \cite{BLMOU23}}. The algorithm takes $O(\log D_T)$ rounds. Then, we can partition the edges of $G$ by connected component (all vertices know their connected component since they know which tree in $T$ they are in).
    Then, since they are independent, it suffices to run in parallel our \MST verification or sensitivity algorithm for each component.
\end{proof}

\subsection{Clusters}
\label{sec:clusters}

\begin{figure}[t]
\label{fig:clusters}
    \begin{subfigure}[t]{0.4\textwidth}
        \centering
        \includegraphics[width=0.8\textwidth,page=1]{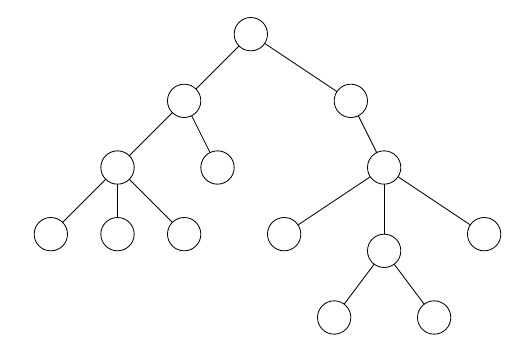}
        \caption{An input tree $T$. Each vertex is its own cluster.}
    \end{subfigure}
    \hfill
    \begin{subfigure}[t]{0.4\textwidth}
        \centering
        \includegraphics[width=0.8\textwidth,page=2]{figures/clusters}
        \caption{$T$ after one cluster contraction step, giving a set of clusters $C_1$. \textcolor{orange}{Orange} polygons denote the \textcolor{orange}{(level-$1$) clusters} that formed.}
    \end{subfigure}
    \vskip\baselineskip
    \begin{subfigure}[t]{0.4\textwidth}
        \centering
        \includegraphics[width=0.8\textwidth,page=3]{figures/clusters}
        \caption{The cluster-tree after one contraction step, which we denote as $T_{C_1}$.}
    \end{subfigure}
    \hfill
    \begin{subfigure}[t]{0.4\textwidth}
        \centering
        \includegraphics[width=0.8\textwidth,page=4]{figures/clusters}
        \caption{A second cluster contraction step, forming some \textcolor{blue}{level-$2$ clusters}. There are four clusters after this step.}
    \end{subfigure}
    \caption{A depiction of our hierarchical clustering.}
\end{figure}

We now introduce our notion of \emph{clusters}, which we use to ``compress'' our input tree $T$ so that we can run algorithms that require extra global memory per vertex.
We give a depiction of our clustering approach in \Cref{fig:clusters}, and they are formally defined as follows:

\begin{definition}[Clusters]
\label{def:clusters}
    A \emph{cluster} is a set of vertices whose induced subgraph in $T$ is connected.
    The \emph{leader} of a cluster is the root of the subtree induced by the vertices in it.
    By $\ell(v)$ we denote the leader of the cluster containing $v$ (and so $\ell(u) = \ell(v)$ iff $u$ and $v$ are in the same cluster).

    We denote the set of clusters at some point in time as $C$.
    Initially, each vertex forms a singleton cluster: i.e.,~$\ell(v) = v$.
    We maintain as an invariant that the union of the clusters at any time is $V$; i.e.,~the vertices are partitioned into clusters.

    Two clusters $c_1$ and $c_2$ are connected if some $u \in c_1$ and $v \in c_2$ are connected by an edge in $T$.
    Note that there can be at most one such edge, for each pair of clusters.
    
    If $T$ is rooted, the graph on clusters forms a rooted \emph{cluster-tree} $T_C$, with ancestor/descendant relationships being defined based on the cluster leaders.

    We say that a vertex $v$ is a \emph{leaf} of a cluster $c$ if it is either a leaf of $T$, or the parent in $T$ of a root of a different cluster.
\end{definition}

Next, we introduce \emph{cluster contraction}, the process by which clusters are created.

\begin{definition}[Cluster Contraction, Sub-Clusters]
\label{def:cluster_contraction}
    Let $c \in C$ be a cluster and let $c_1 \dots c_k$ be some non-empty subset of its child clusters in $T_C$.
    
    When we \emph{contract} $c$, we make a new cluster $c'$ comprised of all vertices in $c, c_1, \dots c_k$.
    The leader of $c'$ will be the leader of $c$.
    We call $c, c_1 \dots, c_k$ the \emph{sub-clusters} of $c'$: in particular we call $c$ the \emph{senior sub-cluster} and $c_1 \dots c_k$ the \emph{junior sub-clusters}.
\end{definition}

\begin{definition}[Contraction Step]
\label{def:contraction_step}
    During our algorithm, we will perform the process described in \Cref{def:cluster_contraction} in parallel for a subset of clusters, ensuring that for all $c_1, c_2, c_3 \in C$, we do \emph{not} simultaneously contract both $c_1$ into $c_2$ and $c_2$ into $c_3$.

    This parallel contraction of a subset of clusters is called a \emph{contraction step}.
\end{definition}

We recall that, as an immediate consequence of \cite{CoyC23} (derandomizing a result of \cite{behnezhad_connectivity}), given a rooted tree, we can contract a set of edges in $O(1)$ rounds to reduce the number of vertices by a constant fraction:

\begin{lemma}[{\cite[Theorem~3.3]{CoyC23}}]
\label{lem:reduce_constant_fraction}
    Given a rooted tree $T$ with $n$ vertices, there is a deterministic, $O(1)$-round \MPC algorithm which contracts a set of edges, reducing the number of vertices in $T$ to $0.99n$.
    
    The algorithm uses $\lspace = O(n^\delta)$ local memory (for some constant $\delta$), and $\gspace = O(m+n)$ global memory.
\end{lemma}

We note that given a cluster-tree, \Cref{lem:reduce_constant_fraction} performs a contraction step as in \Cref{def:contraction_step}.

Next, we formalize the notion of a hierarchical clustering, a series of contraction steps which we perform in our \LCA, \MST verification, and \MST sensitivity algorithms.

\begin{definition}[Hierarchical Clustering]
\label{def:hierarchical_clustering}
    When we perform $k$ contraction steps according to \Cref{lem:reduce_constant_fraction}, we call clusters formed during the $i$-th $(i \in [k])$ step \emph{level-$i$ clusters}.

    We denote the set of all clusters present after contraction step $i$ as $C_i$ (note that not all of these clusters will necessarily be level-$i$ clusters). 
    The singleton clusters at the beginning of the algorithm are all level-$0$ clusters.

    We denote the cluster in $C_i$ which contains $v$ as $c_i(v)$; the leader of this cluster is denoted as $\ell_i(v)$.
\end{definition}

Finally, we observe that there are a linear number of clusters in total, and the total size of all cluster-trees is linear.

\begin{observation}[Size of Hierarchical Clustering]
\label{obs:hierarchical_clustering_size}
    In a hierarchical clustering of the graph as in \Cref{def:hierarchical_clustering}, there are at most $O(n)$ clusters, across all levels.

    Additionally, if each contraction step reduces the number of vertices by a constant fraction (as \Cref{lem:reduce_constant_fraction} does), then the total size of the cluster-trees across all contraction steps is $O(n)$.
\end{observation}

\subsection{All-edges LCA and Ancestor-Descendant Graphs}
\label{sec:all_edges_lca}

Both of our problems take an input graph $G$ as well as a tree $T$. We can assume that $G$ is connected.

We will find it useful for our verification and sensitivity algorithms to work with \emph{ancestor-descendant graphs}: that is, input instances for which all edges in $E \setminus T$ are between ancestors and descendants in $T$.
Formally, we define them as follows:

\begin{definition}[Ancestor-Descendant Graphs]
\label{def:ancestor-descendant-graph}
    Let $G=(V, E)$ be a graph and let $T \subseteq E$ be a rooted spanning tree of $G$.
    $G$ is said to be an \emph{ancestor-descendant graph with respect to $T$} if all edges in $E \setminus T$ are between vertices with an ancestor-descendant relationship in $T$.
\end{definition}

In order to transform our input into an ancestor-descendant graph, we split each non-tree edge $\{u, v\}$ into two edges $\{u, \LCA(u, v)\}$ and $\{v, \LCA(u, v)\}$ with the same weights, where $\LCA(u, v)$ is the lowest common ancestor of $u$ and $v$ in $T$.
We are not able to compute the \LCA of every pair of vertices, since this requires $\Theta(n^2)$ global memory.
Instead, we compute the \LCA for each pair of endpoints of an edge in $E \setminus T$, which requires $\Theta(m)$ memory.
This computation is the subject of the rest of this section.

First, we remark that we can compute the DFS numbering of a rooted tree using a result of Andoni, Stein, and Zhong:

\begin{theorem}[{\cite[Theorem~3]{andoni-biconnectivity}} and {\cite[Theorem~D.22]{arxiv18}}] 
\label{thm:dfs}
    Given a rooted tree $T$ with diameter $D$, there is an $O(\log D)$-round \MPC algorithm which finds the DFS sequence of $T$.
    The algorithm uses $\lspace = O(n^\delta)$ local memory (for some constant $\delta$), and $\gspace = O(n)$ global memory.
    {It succeeds with high probability.}
\end{theorem}

We extend the DFS sequence to a DFS \emph{interval} labeling, which we define as follows:

\begin{definition}
    A DFS interval labeling of a rooted tree $T$ is, for each vertex $v \in T$, an interval $I(v) = [v_{low}, v_{high}]$, where $v_{low}$ is the smallest DFS number in the subtree of $T$ rooted at $v$, and $v_{high}$ is the largest DFS number in the subtree of $T$ rooted at $v$.

    Observe that iff $u$ is an ancestor of $v$ in $T$, then $u_{low} \leq v_{low} \leq v_{high} \leq u_{high}$. In other words, $I(v) \subset I(u)$. If $u$ and $v$ are not in an ancestor-descendant relation in $T$, then $I(u) \cap I(v) = \emptyset$.
\end{definition}

We observe that this interval labeling can be computed using recently developed dynamic programming techniques for the \MPC model:

\begin{lemma}
\label{lem:dfs_interval_labeling}
    Given a rooted tree $T$ with diameter $\dt$, there is an $O(\log \dt)$-round \MPC algorithm which produces a DFS interval labeling of $T$.
    The algorithm uses $\lspace = O(n^\delta)$ local memory (for some constant $\delta$), and $\gspace = O(m+n)$ global memory. It succeeds with high probability.
\end{lemma}
\begin{proof}
    \Cref{thm:dfs} can be applied to produce a DFS sequence of $T$ in $O(\log \dt)$ rounds with global memory $O(m+n)$ and succeeds with high probability.
    By definition, $v_{low}$ is the DFS number of $v$, so it remains for each vertex $v$ to learn $v_{high}$.
    
    This is easily done as a consequence of arguments in \cite{dp_trees_mpc}: the authors mention computing the maximum of input labels in each subtree as a computation which their approach allows.
    Their computation requires $O(\log \dt)$ rounds, deterministically; overall our algorithm succeeds with high probability (since the DFS numbering succeeds only with high probability).
\end{proof}

The goal of the remainder of this section is to prove the following claim:

\begin{theorem}
\label{thm:all_edges_lca}
    Given a graph $G$, and a rooted tree $T$ with diameter $\dt$, there is an $O(\log \dt)$-round \MPC algorithm to compute, for each edge $\{u, v\} \in E \setminus T$ in parallel, $\LCA(u, v)$.
    
    The algorithm uses $\lspace = O(n^\delta)$ local memory (for some constant $\delta$), and $\gspace = O(m+n)$ global memory. It succeeds with high probability.
\end{theorem}

{It seems that the above theorem follows from \cite[Section 4.3]{andoni-biconnectivity}. However, we give the proof explicitly here to provide context for the following sections. Though our proof of \Cref{thm:all_edges_lca} and the proof of \cite{{andoni-biconnectivity}} are fundamentally similar, the details are different.} 

\begin{algorithm}[H]
\caption{$\mathsf{FindLCAClusters}(C)$: Finds, for each edge in $E \setminus T$, the cluster containing the \LCA of their endpoints.}
\label{alg:assign_candidate_lcas}
\begin{algorithmic}[1]
    \For{each edge $\{u, v\} \in E \setminus T$ in parallel}
        \If{$I(c(u)) \cap I(c(v)) \not= \emptyset$} $\LCA_C(u, v) = I(c(u)) \cup I(c(v))$
        \ignore{
        \If{$I(c(u)) \cap I(c(v)) \not= \emptyset$}
            \State $\LCA_C(u, v) = I(c(u)) \cup I(c(v))$
        }
        \Else
        \State set $\chi \gets c(u)$
        \Comment $\chi$ is the ``candidate'' \LCA cluster
        \For{$i \in (\dt, \dt/2, \dots, 4, 2, 1)$}
            \ignore{
            \If{$I(p_i(\chi)) \cap I(p_i(c(v))) = \emptyset$}
                \State $\chi \gets p_i(\chi)$
            \EndIf
            }
            \IfThen{$I(p_i(\chi)) \cap I(p_i(c(v))) = \emptyset$}{$\chi \gets p_i(\chi)$}
        \EndFor

        \State $\LCA_C(u, v) = p(\chi)$
    \EndIf
    \EndFor
    
\State\Return $\LCA_C$
\end{algorithmic}
\end{algorithm}

\begin{algorithm}[H]
\caption{$\mathsf{UndoClustering}(C, \LCA_C)$}
\label{alg:undo_contractions}
\begin{algorithmic}[1]
    \For{$i \in (O(\log \dt)\dots 2, 1)$}
        \For{each level-$i$ cluster $c \in C$ in parallel}
            \State let $c_1\dots c_k$ be the junior sub-clusters of $c$; let $c_p$ be the senior sub-cluster of $c$
            \For{each edge $\{u, v\} \in E \setminus T$ in parallel with $\LCA_C(u, v) = c$}
                \ignore{
                \If{there exists a $c_i$ such that $D(u) \in I(\ell(c_i))$ and $D(v) \in I(\ell(c_i))$}
                    \State $\LCA_C(u, v) \gets c_i$
                \Else
                    \State $\LCA_C(u, v) \gets c_p$
                \EndIf
                }
                \IfThenEElse
                    {there is a $c_i$ such that $D(u) \in I(\ell(c_i))$ and $D(v) \in I(\ell(c_i))$}
                    {$\LCA_C(u, v) \gets c_i$}
                    {$\LCA_C(u, v) \gets c_p$}
            \EndFor
            \State remove $c$ from $C$ and add $c_1\dots c_k$ and $c_p$ to $C$
        \EndFor
    \EndFor
    \ForDo{each edge $\{u, v\} \in E \setminus T$}{$\LCA(u, v) \gets \ell(\LCA_C(u, v))$}
    \ignore{
    \For{each edge $\{u, v\} \in E \setminus T$}
        \State $\LCA(u, v) \gets \ell(\LCA_C(u, v))$
    \EndFor
    }

\State\Return \LCA
\end{algorithmic}
\end{algorithm}

We start by showing that we can create, for each vertex $u$ in a tree in parallel, edges to the parent of $u$, the $2$-parent of $u$, the $4$-parent of $u$, and so on (as in lines 3--5 of \Cref{alg:lca}):

\begin{lemma}
\label{lem:pointer_jumping}
    Let $T$ be a tree with $n$ vertices and diameter $\dt$.
    In $O(\log \dt)$ rounds, we can create edges from each vertex $v$ to $p_1(v), p_2(v), p_4(v), p_8(v)\dots p_n(v)$.
    This requires $\gspace = O(n \cdot \log \dt)$ global memory.
\end{lemma}
\begin{proof}
    Each vertex $v$ starts with a single edge $(v, p(v))$. Machines start by appending a $1$ to these tuples (our example becomes $(v, p(v), 1)$), and creating a copy of the edge with the direction reversed and with a $-1$ appended (in our example, $(p(v), v, -1)$.
    
    Then all edges are sorted lexicographically across machines in $O(1)$ rounds; this leaves, for each vertex $v$, the parent and children of $v$ on consecutive machines.
    Machines can now locally create edges between the parent of $v$ and the children of $v$, and append $2$ to the resulting tuple.
    
    The tuples ending with $-1$ can then be set aside and the process can be repeated for the edges at distance $2$. This process can be repeated $O(\log \dt)$ times (each iteration using edges of distance $i$ to create edges of distance $2i$), and clearly this creates edges from $v$ to $p_1(v)$, $p_2(v)$, $p_4(v)$, and so on; and requires $O(\log \dt)$ rounds and global memory per vertex.
\end{proof}

\begin{algorithm}
\caption{$\LCA(V, E, T)$: Computes the LCA in $T$ for all edges in $E \setminus T$ in parallel.}
\label{alg:lca}
\begin{algorithmic}[1]
    \State Apply \Cref{lem:dfs_interval_labeling} to find a DFS interval labeling of $T$. For each $v \in V$, let $I(v)$ be the interval corresponding to $v$.

    \State Find a hierarchical clustering using \Cref{def:hierarchical_clustering}; let $C$ be the resulting set of clusters.

    \For{$i \in (1, 2, 4, 8, \dots, \dt)$}
    \Comment{{\small Create ``auxiliary'' edges between each cluster in $C$ and their parent in $T_C$, $2$-parent in $T_C$, $4$-parent in $T_C$, and so on}}
        \ignore{
        \For{each cluster $c \in C$ in parallel}
            \State for each cluster $c'$ with $p_i(c') = c$, set $p_{2i}(c') \gets p_i(c)$
        \EndFor
        }
        \ForDo{each cluster $c \in C$ in parallel}
            {for each cluster $c'$ with $p_i(c') = c$, set $p_{2i}(c') \gets p_i(c)$}
    \EndFor

    \State $\LCA_C \gets \mathsf{FindLCAClusters}(C)$
    \Comment{{\small We find the highest $i$ such that $I(p_i(v)) \cap I(p_i(u)) = \emptyset$. This corresponds to the cluster which contains the \LCA of $u$ and $v$}}

    \State $\LCA \gets \mathsf{UndoClustering}(C, \LCA_C)$
    \Comment{{\small We reverse the contraction process and, for each edge in the original graph, we find the vertex which is the \LCA of its endpoints}}
\end{algorithmic}
\end{algorithm}

Next, we argue the correctness and implementation of \Cref{alg:assign_candidate_lcas}: the algorithm which finds, for each non-tree edge, the cluster in which the \LCA of their endpoints is contained.

\begin{lemma}
\label{lem:lca_candidate_lca}
    \Cref{alg:assign_candidate_lcas} can be implemented in $O(\log \dt)$ rounds of sublinear \MPC with $O(m+n)$ global memory.
    After its execution, the cluster containing the \LCA in $T_C$ of $c(u)$ and $c(v)$ for each non-tree edge $\{u, v\} \in E \setminus T$ is correctly computed.
\end{lemma}
\begin{proof}
    Consider some non-tree edge $\{u,v\} \in E \setminus T$.

    First, note that line 2 ensure that if $c(u)$ and $c(v)$ have an ancestor-descendant relationship in $T_C$, then whichever is the ancestor is selected as $\LCA(u, v)$.
    Therefore, assume that this is not the case.

    It suffices to find the smallest $i$ such that the $i$-th ancestor of (w.l.o.g.) $c(u)$ is also an ancestor of $v$: we argue that lines 4--8 accomplish this.
    The approach taken in the for-loop in lines 5--7 moves the candidate $\chi$ for the \LCA up as much as possible such that $\chi$ is not an ancestor of $c(v)$: therefore at the end of the for-loop, the parent of $\chi$ is the \LCA in $T_C$ of $c(u)$ and $c(v)$.

    It remains to show that \Cref{alg:assign_candidate_lcas} can be implemented in \MPC with the specified memory and round requirements.
    First, note that as a precondition to this algorithm auxiliary edges from each cluster to their parent in $T_C$, $2$-parent in $T_C$, $4$-parent in $T_C$, etc.\ were established.
    Suppose (w.l.o.g) the information about the $i$-parent of $c(u)$ is stored as an $(u, i, p_i(u))$ triple.
    We store the relevant information about each non-tree edge as $(\chi, c(v), I(\chi), I(c(v))$ tuples.
    In each iteration of the for-loop we sort the $i$-parent triples along with the candidate tuples, so that the first element of each are stored on consecutive machines.
    We use prefix sum to ensure that a copy of each $i$-parent triple of $u$ is on each machine which has tuples where $\chi=c(u)$.
    With these conditions, machines can perform the necessary computations on line 6  locally.
    The computation takes $O(\log \dt)$ rounds, and the total memory is linear.
\end{proof}

Next, we argue the correctness and implementation of \Cref{alg:undo_contractions}: the algorithm which undoes the contraction steps and finds the \LCA{}s for all non-tree edges.
\begin{lemma}
\label{lem:lca_unwinding}
    \Cref{alg:undo_contractions} can be implemented in $O(\log \dt)$ rounds of sublinear \MPC with $O(m+n)$ global memory.
    After its execution, the \LCA of the endpoints of each non-tree edge $\{u, v\} \in E \setminus T$ is correctly computed.
\end{lemma}
\begin{proof}
    First, we argue that for each non-tree edge $\{u, v\} \in E \setminus T$, $\LCA_C(u, v)$---the \LCA of $c(u)$ and $c(v)$ in $T_C$, computed in \Cref{alg:assign_candidate_lcas}---contains the LCA of $u$ and $v$ in $T$, i.e, that $c(\LCA(u, v)) = \LCA_C(u, v)$.

    Note that for any vertex $u$, $\ell(u)$ (the root of the cluster containing $u$) is an ancestor of $u$ in $T$.
    If $\ell(u)$ is an ancestor of $\LCA(u, v)$ then $c(\LCA(u, v)) = c(u)$ and line 2 of \Cref{alg:assign_candidate_lcas} would return $\LCA_C(u, v) = c(\LCA(u, v))$ as required.
    Suppose otherwise that $\ell(u)$ and $\ell(v)$ are both descendants of $\LCA(u, v)$.
    Recall that \Cref{alg:assign_candidate_lcas} correctly computes the $\LCA$ in $T_C$ of $c(u)$ and $c(v)$ (\Cref{lem:lca_candidate_lca}).
    Therefore there is no $c'$ which is an ancestor of $c(u)$ and $c(v)$ as well as a descendant of $\LCA_C(u, v)$, and it follows that $c(\LCA(u, v)) = \LCA_C(u, v)$.

    We now show that \Cref{alg:undo_contractions} finds $\LCA(u, v)$.
    Each of cluster contraction steps is undone in reverse order: a cluster $c$ produces its senior sub-cluster ($c_p$), and the junior sub-clusters $c_1\dots c_k$.
    In each iteration, for each non-tree edge $\{u, v\} \in E \setminus T$, we check (line 5) whether any of the roots of the junior sub-clusters are a common ancestor of $u$ and $v$.
    If this is the case for some child cluster $c_i$ then by arguments as above $c_i$ must contain $\LCA(u, v)$, and so we set the candidate \LCA of $u$ and $v$ ($\LCA_C(u, v)$) to $c_i$ (line 5).
    If this is not the case, then by elimination $\LCA(u, v)$ must be in the senior sub-cluster, and so we set $\LCA_C(u, v)$ to $c_p$.
    One can see that this finds a vertex which is an ancestor of both $u$ and $v$ and that it finds the lowest such vertex (since the child clusters are checked first).
    After all of the cluster contraction steps have been undone, $C = V$. Then, there is only one possible value for $\LCA(u, v)$, and lines 11 set this value.

    Finally, each iteration of the for-loop in lines 1--10 can be implemented using a constant number of sorts and prefix-sums (sorting the interval labels of the roots of the junior sub-clusters along with the non-tree edge endpoints, and using prefix-sum to determine to which junior sub-cluster $u$ and $v$ belong), and therefore \Cref{alg:undo_contractions} takes $O(\log \dt)$ rounds.
    It is easy to see that the global memory used is linear.
\end{proof}

Finally, we show the main result of this section:

\begin{proof}[Proof of \Cref{thm:all_edges_lca}]
    Correctness, and the implementation of lines 8 and 9 follows immediately from \Cref{lem:lca_candidate_lca,lem:lca_unwinding}.
    Implementation of lines 3--5 follows from \Cref{lem:pointer_jumping}.
\end{proof}

By simple applications of sorting and prefix sum (and with only a constant factor increase in global memory), we can use \Cref{thm:all_edges_lca} to turn our input graph into an ancestor-descendant graph:

\begin{corollary}
\label{cor:break_up_edges}
    Given a graph $G$, and a rooted tree $T$ with diameter $\dt$, there is an $O(\log \dt)$-round \MPC algorithm which replaces each non-tree edge $\{u, v\} \in E \setminus T$ with two edges $\{u, \LCA(u, v)\}$ and $\{v, \LCA(u, v)\}$.

    The algorithm uses $\lspace = O(n^\delta)$ local memory (for some constant $\delta$), and $\gspace = O(m+n)$ global memory. {It succeeds with high probability.}
\end{corollary}

Finally, we observe that the transformation of the input given by \Cref{cor:break_up_edges} does not affect the computation of verification or sensitivity:

\begin{observation}[\cite{Pettie05, DRT92}]
\label{obs:ancestor_descendant_edges_suffice}
    Let $G = (V, E)$ be an edge-weighted graph, and let $T \subseteq E$ be a candidate \MST.

    Replacing each non-tree edge $\{u, v\} \in E \setminus T$ with edges $\{u, \LCA(u, v)\}$ and $\{v, \LCA(u, v)\}$ (of the same weight) does not affect the result of \MST verification, or the sensitivity of tree edges.
    After the replacements, the sensitivity of non-tree edges is equal to the minimum sensitivity of the two replacements.
\end{observation}

\section{\MST Verification}
\label{sec:mst_verification}

In this section we present our algorithm for \MST verification, and prove \Cref{thm:mst_verification}.

\begin{theorem}[\textbf{\MST Verification Algorithm}]
\label{thm:mst_verification}
Let $\delta$ be an arbitrarily small positive constant. Let $G = (V, E)$ be an edge-weighted input graph and let $T \subseteq E$ be a  tree. Suppose the diameter of $T$ be $\dt$. Then one can decide with high probability whether $T$ is an \MST of $G$ in $O(\log \dt)$ rounds on an \MPC with local memory $\lspace = O(n^\delta)$ and optimal global memory $\gspace = O(m+n)$.
\end{theorem}

We recall that, due to \Cref{rem:input_is_spanning_tree}, $T$ can be assumed to be a rooted spanning tree of $G$, because we can verify this in $O(\log \dt)$ rounds. From \Cref{rem:val-D-T}, we can assume that the algorithm knows the value of $D_T$. Also, recalling \Cref{rem:extend_to_forests}, our \MST verification algorithm, while stated for trees, can easily be extended to forests. We may also assume that $G$ is an ancestor-descendant graph with respect to $T$ (\Cref{def:ancestor-descendant-graph}). If this is not the case then we can use \Cref{cor:break_up_edges} to transform our input into an ancestor-descendant graph, with asymptotically the same number of edges. By \Cref{obs:ancestor_descendant_edges_suffice}, this transformation does not affect the output of the \MST verification problem.

A naive algorithm for the \MST verification problem might be to collect, for each vertex $u$, the path from $u$ to the root of $T$.
Using prefix-sum in $O(1)$ rounds, we can then compute, for each vertex $u'$ on the path from $u$ to $T$, the highest-weight edge on the path from $u$ to $u'$.
Then, for each non-tree edge $e = \{u, v\}$ in parallel (where $v$ is an ancestor of $u$ in $T$), we can use sorting and prefix-sum and the previously computed prefix-sum of the path from $u$ to the root to compute the highest weight of an edge on the path from $u$ to $v$.
Suppose this weight is $x_e$: if $x_e < w(e)$ (and this is true for all non-tree edges), then $T$ is a minimum spanning tree; otherwise it is not.

There is a problem with this algorithm, which is that to collect the path to the root for all vertices might require $\Omega(n \cdot \dt)$ global memory, which might be $\omega(m+n)$.
If we were able to somehow reduce the number of vertices in the graph to $\frac{n}{\dt}$, then this approach would work.

The challenge is therefore to preserve the information required to evaluate for each non-tree edge $\{u, v\}$, the highest weight of an edge from $u$ to $\LCA(u, v)$, while reducing the number of vertices by a factor of $D$. The  (effective) number of vertices/clusters reduces by a constant factor in each round by applying contraction. If we forget about the weights of the tree edges between contracted vertices, we may lose relevant information. 

Initially, we have a set of $n$ clusters where each cluster is a singleton vertex and over rounds these clusters are contracted. We need to preserve enough information about the weights of edges in the clusters to check whether any edge in $E \setminus T$ ought to be in a minimum spanning tree.
To do this, we define the notion of a \emph{weight preserving labeling} to keep track of two labels throughout the cluster contraction process (as formally defined in \Cref{defi:label}).
We apply the contraction process on $T$ for $O(\log D_T)$ rounds until we have a tree of $O(n/D_T)$ clusters.
Then, we can find the path from every cluster to the root cluster using $O(n)$ global space (\Cref{lem:collecting_paths}).
Finally, we argue that this path information along with the weight preserving labeling of clusters that we have maintained, is enough information to decide whether $T$ is an MST. 

\subsection{Weight-Preserving Labeling}
In this subsection we introduce the notion of a \emph{weight-preserving labeling}. As already discussed, the intuition is that we need to reduce the number of vertices by a factor of $\Omega(\dt)$ so that subsequent steps of the algorithm have $\Omega(\dt)$ global memory per vertex to work with. We do this using cluster contraction steps as outlined in \Cref{def:cluster_contraction}.

\begin{definition}[Weight-Preserving Labeling]\label{defi:label}
    Let $G = (V, E)$ be a graph with an edge-weighting function $w : E \rightarrow [W]$, $T \subset E$ be a candidate \MST, and $C$ be a partition of $V$ into a set of clusters.
    Assume that $G$ is an ancestor-descendant graph with respect to $T$.
    A \emph{weight-preserving labeling} of $(G, w, T, C)$ is a pair $(\through, \out)$ such that:
    \begin{itemize}
        \item for each edge $\{c_c, c_p\} \in T_C$ (where $c_p$ is the parent of $c_c$ in $T$), a label $\through(c_c, c_p)$ which takes a value in $[W]$; and
        \item for each edge $\{u, v\} \in E \setminus T$, labels $\out(u,v)$ and $\out(v, u)$ which take values in $[W] \cup \{-\infty\}$.
    \end{itemize}

    These labels represent the following:
    \begin{itemize}
        \item For each edge $\{c_p, c_c\} \in T_C$, let $c_p$ be the parent cluster and let $c_c$ be the child cluster.
        $\through(c_p, c_c)$ is the highest weight of an edge on the path from $\ell(c_p)$ to $p(\ell(c_c))$ in $T$.
        
        \item For each edge $\{u, v\} \in E \setminus T$, let $P$ be the path from $u$ to $v$ in $T$.
        $\out(u, v)$ is the highest weight edge on $P$  where both endpoints are in $c(u)$; $\out(v, u)$ is defined analogously.

        If there are no such edges, then $\out(u, v) = -\infty$.
        Clearly, if $v \in c(u)$, then $\out(u, v) = \out(v, u)$.
    \end{itemize}
\end{definition}

These labels, along with the weights of some edges which we do not contract, sufficiently capture all information necessary to evaluate the highest weight of an edge in $T$ which some edge $\{u, v\} \in E \setminus T$ covers:

\begin{observation}
\label{obs:weight_preserving_labelings_work}
    Consider a graph $G=(V, E)$ with a weighting function $w$, a rooted candidate \MST $T$, and a set of clusters $C$.
    Suppose that $G$ is an ancestor-descendant graph with respect to $T$, and let $(\through, \out)$ be a weight-preserving labeling of $(G, w, T, C)$.

    For any edge $\{u, v\} \in E \setminus T$, let $u$ be an ancestor of $v$ and let $P_C = (c(v)\dots c(u))$ be the path of clusters in $T_C$ which the path from $u$ to $v$ in $T$ passes through.
    Then, the highest weight of an edge on the path from $u$ to $v$ in $T$ is:
    \begin{equation*}
    \max
    \Biggl\{
        \out(u, v), \out(v, u), 
        \max_{(c_c, c_p) \in P_C | c_p \not=c(u)} 
        \left\{\through(c_p, c_c)\right\}, 
        \max_{(c_c, c_p) \in P_C} \biggl\{\ w\Bigl(\bigl\{\ell(c_c), p(\ell(c_c))\bigr\}\Bigr)\ \biggr\} 
    \Biggr\}
    \end{equation*}
\end{observation}

We now argue that when we perform cluster contraction steps, we are able to efficiently maintain a weight-preserving labeling of the current clustering:

\begin{lemma}
\label{lem:cpl_recomputation}
    Consider a graph $G=(V, E)$ with a weighting function $w$, a rooted tree $T$, and a set of clusters $C$.
    Suppose that $G$ is an ancestor-descendant graph with respect to $T$, and let $(\through, \out)$ be a weight-preserving labeling of $(G, w, T, C)$.
    
    After a cluster contraction step is performed, returning a new set of clusters $C'$, we can compute a new weight-preserving labeling $(\through', \out')$ of $(G, w, T, C')$ using only $\through$, $\out$, and the weights of edges between vertices in different clusters in $C$.
\end{lemma}

\begin{proof}
    First, note that if some cluster $c \in C$ is also a cluster in $C'$ (i.e.,~$c$ is not involved in a contraction), then for all vertices $u \in c$ and all non-tree edges $\{u, v\}$, $\out(u, v) = \out'(u, v)$.
    Also, for each child $c_i$ of $c$ in $C$, $\through(c_i, c) = \through'(c*_i,c)$, where $c*_i$ is the cluster which the vertices in $c_i$ are part of in $C'$.

    So we focus on clusters which have undergone a contraction step.
    Consider a cluster $c \in C$: suppose it has children $c_1\dots c_k$ in $T_C$, and that it forms cluster $c' \in C'$ when contracted.

    Let $c_i'$ be a child cluster of $c'$ in $T_{C'}$. 
    We handle each type of label separately and show that they can be properly computed:
    \begin{itemize}
        \item We show that we can compute $\through'(c_i', c')$.
        Recall that $\through'(c_i', c')$ is the highest-weight edge on a path in $T$ from the leader of $c'$ to the parent of the leader of $c_i'$.
        
        Since $c_i'$ is a child of $c'$ in $T_{C'}$, it must have been a child of some $c_i$ in $T_C$.
        The path from the leader of $c$ (also the leader of $c'$) to the leader of $c_i'$ must pass through $c_i$, and it suffices to take $\through'(c_i',c') = \max(\through(c_i', c_i),$ $w(\{\ell(c_i), p(\ell(c_i))\}) \through(c_i, c))$.
        
        \item We show that we can compute $\out'(u, v)$.
        Let $\{u, v\} \in E \setminus T$ and suppose $u \in c'$. Let $P$ be the path from $u$ to $v$ in $T$.
        Recall that $\out'(u, v)$ is the maximum-weight edge on $P$ where both endpoints are in the cluster in $C'$ containing $u$.
        
        Let $c_i$ be an arbitrary child of $c$ in $T_C$.
        There are $5$ sub-cases here based on the positions of $u$ and $v$ w.r.t. clusters that we are contracting:
        \begin{enumerate}
            \item $v \in c'$: In this case, either $u$ and $v$ were in the same cluster in $C$ (in which case $\out(u, v) = \out'(u, v)$); or $u$ and $v$ were in different clusters in $C$.
            In the latter case note our assumption that all non-tree edges are ancestor-descendant, so suppose w.l.o.g.\ that $u$ was in $c$ and $v$ was in a child $c_i$ of $c$.
            Then, since all but one edge had endpoints either both in the cluster containing $u$ or both in the cluster containing $v$, $\out'(u, v) = \max\{\out(u, v), w(\ell(c_i), p(\ell(c_i)), \out(v, u)\}$.
           
            \item $u \in c$, and $P$ does not go through any child of $c$ in $T_C$: In this case $P$ must leave $c$ via its leader.
            This means that the section of the path in the cluster containing $u$ contains the same edges, and so we have $\out'(u, v) = \out(u, v)$.
            
            \item $u \in c$, and $P$ goes through cluster $c_i$ in $T_C$. Since $c' \in C'$ contains all vertices in both $c$ and $c_i$, $\out'(u, v)$ must represent the highest-weight encountered when leaving $c$, taking the edge between $c$ to $c_i$, and then leaving $c_i$ via one of its child clusters (say ${c_i}_j$) in $C$. Then $\out'(u, v) = \max\{\out(u, v), w(\ell(c_i), p(\ell(c_i))), \through({c_i}_j, c_i)\}$.
            
            \item $u \in c_i$, and $P$ does not go through $c$ in $T_C$. 
            In this case $P$ must leave $c_i$ via one of its children in $C$. The contraction of $c$ with its children does not affect the edges encountered when leaving $c_i$ via one of its children, and so $\out'(u, v) = \out(u, v)$.
            
            \item $u \in c_i$, and $P$ goes through $c$ in $T_C$, and exits $c$ via its leader. 
            This case is analogous to case 3: we have $\out'(u, v) = \max \{\out(u, v), w(\ell(c_i), \ell(p(c_i))), \through(c_i, c)\}$.     \qedhere  
        \end{enumerate}
    \end{itemize}
\end{proof}

We argue that we are able to implement the above approach in sublinear \MPC, using optimal global memory:

\begin{lemma}
\label{lem:vertex_reduction_with_wpl}
    Let $G$ be a graph with a set $C$ of clusters and a weight-preserving labeling $(\through, \out)$.
    There is an \MPC algorithm which, in $O(1)$ rounds, outputs a new set of clusters $C'$ such that $|C'| \leq |C| / \alpha$ (for some constant $\alpha > 1$), along with a weight-preserving labeling $(\through', \out')$ for $C'$.
    The algorithm uses $\lspace = O(n^\delta)$ and $\gspace = O(n+m)$.
\end{lemma}
\begin{proof}
    Observe that the computation of $\through'(\cdot, c)$ for a cluster $c$ only depends on $\through(c_i, c)$ and $\through(\cdot, c_i)$ for all children $c_i$ of $c$; this computation can be performed locally along with the contractions.

    Similarly, the computation of $\out'(u, v)$ depends on $\out(u, v)$. It also depends on which case of \Cref{lem:cpl_recomputation} applies, and depending on the case a constant amount of additional information (\through values and the weight of edges which are being contracted), which can be made locally available using sorting.
    The applicable case can be computed locally using the interval labeling of the vertices, and the constant amount of additional information for each non-tree edge can be placed on the same machine as the edge and used to compute $\out'(u, v)$.

    Since there is one \through label for each cluster in $C_i$, and $|C_i| < n$; and there are two \out labels for each non-tree edge, of which there are $m$, the algorithm requires $O(m+n)$ total memory.
\end{proof}

We now argue that we can perform $O(\log \dt)$ cluster contraction steps, while keeping track of a weight-preserving labeling, in $O(\log \dt)$ rounds of sublinear \MPC with optimal global memory.

\begin{corollary}
\label{cor:vertex_reduction}
    Applying the algorithm from \Cref{lem:vertex_reduction_with_wpl} on the input instance $O(\log \dt)$ times gives a graph with a weight-preserving labeling and $\frac{n}{\poly(\dt)}$ clusters.
\end{corollary}

\subsection{The Rest of the Algorithm}

We now show that we can use the weight-preserving labeling from the previous section, along with the approach sketched earlier, to solve the problem of \MST verification.

First, we argue that we can collect, for each vertex $u$ in a tree, the path from $u$ to the root of the tree, in $O(\log \dt)$ rounds with an extra factor of $\dt$ global memory:

\begin{lemma}
\label{lem:collecting_paths}
    Let $T$ be a tree with $n$ vertices and diameter $\dt$.
    In $O(\log \dt)$ rounds, we can, in parallel for each vertex $v$, collect the path from $v$ to the root $r$ of $T$. This uses $\gspace = O(n \cdot \dt)$ global memory.
\end{lemma}
\begin{proof}
A similar procedure can be used as in \Cref{lem:pointer_jumping}, except after each iteration the edges with distance $i$ are not discarded. This still takes $O(\log \dt)$ rounds, but now $\dt$ global memory per vertex is required, since the entire path to the root is being stored.
\end{proof}

Finally, we conclude by proving the main result of the section:

\begin{proof}[Proof of~\Cref{thm:mst_verification}]
    We start by computing a clustering using \Cref{cor:vertex_reduction}, along with a weight-preserving labeling $(\through, \out)$ of $(G, w, T, C)$.

    Note that the diameter of $T_C$ is at most $\dt$.
    We can execute \Cref{lem:collecting_paths} on $T_C$ (taking $O(\log\dt)$ rounds and using $O(|C| \cdot D_T) = O(n)$ global memory), and then for each cluster $c \in C$, the path from $c$ to the root of $T_C$ is stored on consecutive machines.
    We can see that, for each cluster $c$ in the path from $c(u)$ to $r$, we can compute the maximum $\through$ value on the path from $c(u)$ to $c$: this can be done with prefix-sum in $O(1)$ rounds.
    In a similar way, we can compute the maximum inter-cluster edge on the path from $c(u)$ to $c$ (for each of the ancestors of $u$ in $T_C$, $c$), in $O(1)$ rounds.

    For each non-tree edge $\{u, v\} \in E \setminus T$, it suffices to collect $\out(u, v)$, $\out(v, u)$, $w(\{u, v\})$, one of these maximum $\through$ values, and one of these maximum inter-cluster edge values.
    Then, it can be determined locally whether $\{u, v\}$ is of larger weight than all edges on the path from $u$ to $v$ in $T$.
    This information can be collected onto a single machine using $O(1)$ rounds of sorting, after which the necessary computation can be performed locally.
    In total the information required for each non-tree edge $\{u, v\}$ is of constant size.

    Finally, we note that all steps are deterministic, except for the computation of the DFS interval labeling (\Cref{lem:dfs_interval_labeling}): therefore, the algorithm as a whole succeeds with high probability.
\end{proof}

\section{MST Sensitivity}
\label{sec:sensitivity}

In this section we present our \MST sensitivity algorithm.
The conclusion of this section is the following: 

\begin{theorem}[\textbf{\MST Sensitivity Algorithm}]
\label{thm:mst_sensitivity}Let $\delta$ be an arbitrarily small positive constant. Let $G = (V, E)$ be an edge-weighted input graph and let $T \subseteq E$ be an \MST of $G$ with diameter $\dt$.

The problem of \MST sensitivity can be solved in $O(\log \dt)$ rounds (with high probability) on an \MPC with local memory  $\lspace = O(n^\delta)$ and optimal global memory $\gspace = O(m+n)$.
\end{theorem}

We recall that, given \Cref{rem:input_is_spanning_tree} we can assume that $T$ is a rooted spanning tree. From \Cref{rem:val-D-T}, we can assume that the algorithm knows the value of $D_T$. Given \Cref{rem:extend_to_forests}, while our algorithm is stated for trees it can easily be extended to forests; and that, given \Cref{obs:ancestor_descendant_edges_suffice} and \Cref{cor:break_up_edges}, we can assume that $G$ is an ancestor-descendant graph (\Cref{def:ancestor-descendant-graph}) with respect to $T$.
We also recall that we can
compute a DFS interval label for all vertices in $O(\log \dt)$ rounds (\Cref{lem:dfs_interval_labeling}).

First, we note that our \MST verification algorithm works by computing, for each non-tree edge $e$, the highest weight of an edge which $e$ covers in $T$.
Since this is the information which is relevant for computing the sensitivity of non-tree edges, this almost gives us an \MST sensitivity algorithm for non-tree edges immediately.
It remains to compute, for each pair of edges $\{u, \LCA(u, v)\}$, $\{v, \LCA(u, v)\}$, which has the lowest sensitivity.
This is easily accomplished using sorting and prefix-sum in $O(1)$ rounds, and so we have the following:

\begin{observation}
\label{obs:sensitivity_of_non_tree_edges}
The \MST verification algorithm outlined in \Cref{sec:mst_verification} can be extended to compute the sensitivity of edges in $E \setminus T$.
\end{observation}

Finally, we observe that the problem of \MST sensitivity is (for tree edges) the same computing the minimum weight of a covering edge:

\begin{observation}
\label{obs:sensitivity_is_min_covering}
    To compute the sensitivity of a tree edge $e$ ($\sens(e)$, see \Cref{def:mst_sensitivity}), it suffices to compute the minimum weight of a non-tree edge which covers $e$ (see \Cref{def:covering}).
    Let $\mc(e)$ be the minimum weight of a non-tree edge which covers $e$.
    Then $\sens(e) = \mc(e) - w(e)$.
\end{observation}

For ease of notation, we will work with \mc rather than \sens values in this section.
We can set aside a series for machines for keeping track of and updating \mc values: this requires an additive $O(n)$ global memory (and so does not affect the total asymptotically).
These machines can process $O(n)$ updates to the \mc values of vertices in $O(1)$ rounds, using sorting and prefix sum.

We will use an approach to solve the problem of \MST sensitivity which employs a similar method to our \MST verification algorithm: reduce the number of vertices by a factor of $\Omega(\poly(\dt))$ using a hierarchical clustering; use an algorithm which requires an additional $\Omega(\poly(\dt))$ memory per vertex; and finally use this information to compute the sensitivity values for the original edges of the tree.

However, this approach is much more complex than it was for the \MST verification algorithm.
In particular, for verification it sufficed to \emph{summarize} information about the clusters.
For sensitivity we have to be careful to preserve enough information about the clusters so that, when we undo the clustering, we can eventually compute the sensitivity of the edges inside them. Informally speaking, we need to be able to retrieve the information related to clusters from any intermediate step. Since there are $O(\log D_T)$ many intermediate steps before obtaining final clusters, $O(\log D_T)$ rounds are enough to retrieve the information. 

In \Cref{sec:contraction}, we discuss our contraction process for \MST sensitivity. Note that, at the end of the contraction process (as in \Cref{alg:sens:cluster_contractions}), we have a clustering $C$, a tree $T_C$ of $O(n/\poly(\dt))$ clusters, and a graph $G'$ (obtained by replacing some of the original edges).
In \Cref{subsec:simulating_pettie}, we discuss how to find the sensitivity of the edges between adjacent clusters in $C$ (See in \Cref{alg:sens:pettie}).
This is made possible by the extra $\poly(\dt)$ global space we have per cluster.
Finally, in \Cref{sec:unwind}, we discuss how we ``unwind'' the relevant information w.r.t. intermediate clusters in $O(\log D_T)$ rounds to determine the sensitivity of the edges in $T$ which we contracted into clusters during the contraction phase (see \Cref{alg:sens:undo_contractions}).

Our overall algorithm is as follows:

\begin{algorithm}
\caption{$\mathsf{TreeEdgeSensitivity}(G, T)$}
\label{alg:sens}
\begin{algorithmic}[1]
    \State For each edge $e \in T$, set $\mc(e) = \infty$. 
    \State Compute a DFS interval label $I$ for all vertices.
    \State Replace each edge $\{u, v\} \in E \setminus T$ with edges $\{u, \LCA(u, v)\}$ and $\{v, \LCA(u, v)\}$ (\Cref{cor:break_up_edges}).
    \State $C_\tau, N, E' \gets \mathsf{ClusterContractions}(G, T)$ (\Cref{alg:sens:cluster_contractions})
    \State $N' \gets \mathsf{ClusterSensitivity}(G, E', T_{C_\tau}, C_\tau)$ (\Cref{alg:sens:pettie})
    \State $\mathsf{UndoContractions}(G, T, N \cup N', C_\tau)$ (\Cref{alg:sens:undo_contractions})
\end{algorithmic}
\end{algorithm}

\subsection{Contraction into Clusters}\label{sec:contraction}
To run our sensitivity algorithm in \Cref{subsec:simulating_pettie}, we need to have $\poly(\dt)$ global memory per vertex. To accomplish this, we first perform a hierarchical clustering, producing a set of clusters $C$. We then run our sensitivity algorithm on $T_C$.

We are going to want to maintain the invariant that, at each step of the hierarchical clustering, if an endpoint of a non-tree edge is in a cluster $c$, it does not cover any of the edges inside $c$.

When we perform a contraction, we maintain this invariant by ``cutting off'' any ends of a non-tree edge which would violate the invariant. Suppose there is a non-tree edge $e = \{u, v\}$, and after a contraction step $u$ is inside a cluster $c$. Let $u'$ be the last vertex on the path from $u$ to $v$ which is also in $c$. Suppose $u \not= u'$: in this case the invariant would be violated. We keep the part of the edge which is outside of the cluster ($\{u', v\}$) and carry it forward to future contraction steps. For the part of the edge inside the newly-formed cluster $c$, we first observe that (because the invariant was preserved from the previous contraction step) this part forms a path, from the root of one sub-cluster to one of its leaves.

We create what we call a \emph{root-to-leaf note}: a reminder of the part of the edge which we do not carry forward, to be dealt with later in the algorithm. We formally define this note as follows:

\begin{definition}[Root-to-Leaf Note]
\label{def:root_to_leaf_note}
    A \emph{root-to-leaf note} is a tuple $(r, l, i, w)$, where: 
    \begin{itemize}
        \item $r$ is the root of a cluster $c$;
        \item $l$ is one of the leaves of $c$;
        \item $i$ is the index of the contraction step in which the cluster $c$ was formed;
        \item $w$ is a weight.
    \end{itemize}
\end{definition}

Note that if we simultaneously create multiple root-to-leaf notes where $r$, $l$, and $i$ are the same, it suffices to keep only the root-to-leaf note with the minimum $w$.
The intuition behind root-to-leaf notes is that some non-tree edge covers a root-to-leaf path within some (now contracted) cluster. When expanding out the contracted out clusters (see \Cref{alg:sens:undo_contractions}) we ensure that we take the root-to-leaf notes into account when computing the \mc value for all edges which they cover.

\begin{figure}
\label{fig:sensitivity_contraction}
    \begin{subfigure}[t]{0.2\linewidth}
        \centering
        \includegraphics[width=\linewidth,page=1]{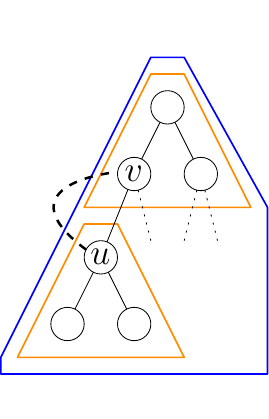}
        \caption{Case 1}
    \end{subfigure}
    \hspace{7mm} 
    \begin{subfigure}[t]{0.2\linewidth}
        \centering
        \includegraphics[width=\linewidth,page=2]{figures/sensitivity_contraction}
        \caption{Case 2}
    \end{subfigure} 
    \hspace{7mm} 
    \begin{subfigure}[t]{0.2\linewidth}
        \centering
        \includegraphics[width=\linewidth,page=4]{figures/sensitivity_contraction}
        \caption{Case 4}
    \end{subfigure}
    \hspace{7mm} 
    \begin{subfigure}[t]{0.2\linewidth}
        \centering
        \includegraphics[width=\linewidth,page=5]{figures/sensitivity_contraction}
        \caption{Case 5}
    \end{subfigure}
    \caption{A depiction of cases 1, 2, 4, and 5 of the sensitivity contraction process (\Cref{def:sensitivity_contraction_process}).
    (Case 3 omitted for brevity.)
    The \textcolor{orange}{orange} clusters $c$ (senior) and $c_1$ (junior) are being contracted to form the \textcolor{blue}{blue} cluster $c'$.
    The dashed edge is the relevant non-tree edge.
    In cases~4 and~5, we replace the non-tree edge with a root-to-leaf note (the \textcolor{OliveGreen}{dashed, green} line) covering part of the non-tree edge inside $c'$, and shorten the remaining edge (the \textcolor{OliveGreen}{solid, green} line).}
\end{figure}

The process of performing cluster contraction steps while maintaining the invariant above is depicted in \Cref{fig:sensitivity_contraction} formally described as follows:

\begin{definition}[Sensitivity Contraction Process]
\label{def:sensitivity_contraction_process}
Assume that we are at contraction step $i$.
When we contract child clusters $c_1 \dots c_k$ into cluster $c$ to form $c'$, for each non-tree edge $e = \{u, v\}$ there are five possibilities:
\begin{enumerate}
    \item $u$ is in one of the junior sub-clusters, and $v$ is in the senior sub-cluster.
    Since $\{u, v\}$ does not cover any edges in clusters by assumption, $u$ must be the root of the junior sub-cluster and $v$ must be one of the leaves of the senior sub-cluster.
    It suffices to set $\mc(\{u, v\}) \gets \min(\mc(\{u, v\}), w(e))$.
    
    \item $u$ is in the senior sub-cluster, and $v$ is an ancestor of $u$.
    By the invariant, $u$ must be the root of the senior sub-cluster.
    In this case the invariant is automatically maintained and no action is necessary.
    
    \item $u$ is in a junior sub-cluster, and $v$ is a descendant of $u$.
    By the invariant, $v$ must be a leaf of a junior sub-cluster.
    In this case the invariant is automatically maintained and no action is necessary.

    \item $u$ is in a junior sub-cluster, and $v$ is an ancestor of $c'$.
    By the invariant, $u$ must be the leader of its junior sub-cluster.
    In this case we replace $e$ with an edge $\{\ell(c'), v\}$ (recall: $\ell(c')$ is the leader of $c'$).
    Since we have shortened the edge, we create a \emph{root-to-leaf note} $(\ell(c'), p(u), i, w(\{u, v\}))$, where $i$ is the index of the contraction step in which $c$ was formed.
    We also set $\mc(\{u, p(u)\}) \gets \min(\mc(\{u, p(u)\}, w(e))$.

    \item $u$ is in a senior sub-cluster, and $v$ is an descendant of one of the junior sub-clusters of $c'$, $c_j$.
    By the invariant, $u$ must be a leaf of its junior sub-cluster.
    In this case we replace $e$ with an edge $\{l, v\}$, where $l$ is the leaf of $c'$ on the path from $u$ to $v$ in $T$.
    Since we have shortened the edge, we create a \emph{root-to-leaf note} $(\ell(c_j), l, i, w(\{u, v\}))$, where $i$ is the index of the contraction step in which $c_j$ was formed.
    We also set $\mc(\{\ell(c_j), u\}) \gets \min(\mc(\{\ell(c_j), u\}), w(e)$.
\end{enumerate}

We note that if both endpoints require the creation of root-to-leaf notes, then we create both.
\end{definition}

We now argue that this only results in $O(n)$ root-to-leaf notes being produced, and so our global memory constraints are not violated:

\begin{lemma}
    Performing the above operation for all non-tree edges in parallel after each cluster contraction step (\Cref{lem:reduce_constant_fraction}) results in a total of $O(n)$ root-to-leaf notes.
\end{lemma}
\begin{proof}
    Follows from \Cref{obs:hierarchical_clustering_size}.
    Fix a cluster $c$ formed in step $i$. Then:
    \begin{itemize}
        \item In case 4, the number of distinct root-to-leaf notes created is at most the number of junior sub-clusters of $c$.
        The total number of junior sub-clusters is at most the number of vertices in $T_{C_i}$: summing this over all contraction steps gives at most $O(n)$ such root-to-leaf notes.
        \item In case 5, the number of distinct root-to-leaf notes created is at most the number of leaves of $c$ which have a child which is a cluster leader.
        Similarly, the total number of these is the number of vertices in $T_{C_i}$: summing over all contraction steps gives at most $O(n)$ such root-to-leaf notes. \qedhere
    \end{itemize} 
\end{proof}

We now argue that we can actually perform this computation while we do a cluster contraction step:

\begin{lemma}
\label{lem:sens:computing_cluster_contractions}
    \Cref{alg:sens:cluster_contractions} can be implemented deterministically on an \MPC with $\lspace = O(n^\delta)$ and $\gspace = O(n+m)$ in $O(\log \dt)$ rounds.
\end{lemma}
\begin{proof}
    It suffices to show that one iteration of \Cref{alg:sens:cluster_contractions} can be implemented in $O(1)$ rounds within the specified memory constraints.
    The edges which are to be contracted in each iteration can be identified in $O(1)$ rounds (by \Cref{lem:reduce_constant_fraction}).
    It remains to show that we can identify for each non-tree edge $e \in E'$ whether case~1, 4, or~5 of \Cref{def:sensitivity_contraction_process} applies.

    Fix some non-tree edge $e \in E'$, and suppose that $e = \{u, v\}$, and that $u$ is an ancestor of $v$.

    \begin{itemize}
    \item It is easily checked whether case $1$ of \Cref{def:sensitivity_contraction_process} applies to $e$, since this is the case if and only if $e$ is the same as one of the edges which are being contracted.

        \item  For case~4, we need to check whether some edge $\{u, u'\}$ is being contracted, where $I(u) \subset I(u')$. If this is the case then case~4 applies.
        We need to set $\mc(\{u, u'\}) = \min\{\mc(\{u, u'\}), \\ w(e)\} $; create a root-to-leaf note $(\ell(u'), u, i, w)$; and truncate $e$ so that it goes between $\ell(u')$ and $v$.
        
        This can be accomplished by sorting the edges in $E'$ along with the edges to be contracted, and using prefix-sum to communicate the edge $\{u, u'\}$ to all edges in $E'$ with $u$ as an endpoint (if any).
        Then, machines can perform the necessary changes to the graph locally (sending the \mc update to the machines dedicated to storing and updating those values).

        \item For case~5, we need to check whether some edge $\{v, v'\}$ with $I(v) \supset I(v')$ is being contracted. If so, then case~5 applies.
        
        There is one added complication here: we need to find the correct leaf for the root-to-leaf note (that is, the leaf of the cluster containing $v'$ through which the $u-v$ path in $T$ passes).
        Since for any cluster $c$ we know what its children in $T_C$ are, we know what the leaves of $c$ are (they are the parents of the roots of the children of $c$ in $T_C$; they cannot be leaves of $T$ since the $u-v$ path could not pass through them).
        So by sorting the DFS interval labels of the leaves of the cluster along with the edges in $E'$ (and using prefix-sum), we can find this leaf: let us call it $l$.

        We need to set $\mc(\{v, v'\}) = \min\{\mc(\{v, v'\}), w(e)\} $; create a root-to-leaf note $(\ell(v), l, i, w)$; and truncate $e$ so that it goes between $l$ and $u$.
    \end{itemize}

    Since all cases can be identified and the correct action taken within $O(1)$ rounds and linear global memory (using constantly many applications of sorting and prefix sum), the lemma result follows.
\end{proof}

Finally, we present the algorithm for this first step as a whole:

\begin{algorithm}
\caption{$\mathsf{ClusterContractions}(G, T)$: Performs $O(\log \dt)$ cluster contraction steps}
\label{alg:sens:cluster_contractions}
\begin{algorithmic}[1]
    \State Let $C_0 \gets V$; $E' \gets E \setminus T$

    \For{$i \in (1, 2, \dots, \tau = O(\log \dt))$}
        \State Perform an iteration of \Cref{def:sensitivity_contraction_process} on $T_{C_{i-1}}$, using \Cref{lem:reduce_constant_fraction} to choose the edges to contract.
        
        \State $C_i \gets$ The set of clusters after this iteration.

        \State $N_i \gets$ The \emph{root-to-leaf} notes created during this iteration.

        \State $E' \gets$ The truncated non-tree edges (from $E \setminus T$)
    \EndFor

\State\Return $C_\tau$: The clusters after $\tau = O(\log \dt)$ cluster contraction steps

\State\Return $N = N_1 \cup N_2 \cup \dots \cup N_\tau$: The root-to-leaf notes created by the cluster contraction steps

\State\Return $E'$: The truncated non-tree edges (from $E \setminus T$)
\end{algorithmic}
\end{algorithm}

\subsection{\MST Sensitivity with $n/\poly(D)$ Vertices}
\label{subsec:simulating_pettie}

After executing \Cref{alg:sens:cluster_contractions}, we have a clustering of the vertices in $V$ into a set of $n / \poly(D)$ clusters $C_\tau$.
We note that the diameter of $T_{C_\tau}$ is at most $\dt$: for simplicity here we will assume that it \emph{is} $\dt$.

We will now present an algorithm for \MST sensitivity, which works in $O(\log \dt)$ rounds but with $\poly(D)$ space per node.
We will run this algorithm on the cluster tree $T_{C_\tau}$, and use the result to compute the sensitivity of edges in $T$.
Our algorithm works in a similar way to the algorithm of Pettie \cite{Pettie15} in the centralized setting.

In this section, we denote the depth of a cluster $v$ in $T_C$ by $\lev(v)$: it is defined as the distance of $v$ from the root cluster $r$.

We first define a useful array for our algorithm:

\begin{definition}\label{def:array}
    For each non-root cluster $c \in C_\tau$, $A_c$ is an array of length $\dt$ such that for each $i \in [\dt]$
\[   A_c[i]=\left\{
\begin{array}{l l}
 \min_{\{\{u, v\} \in E'': c(u) = c \wedge c(v) = x\}}   \{w(\{u,v\}):\lev(x)<i 
 \}  &1 \leq i \leq \lev(c) \\
    \infty & i > \lev(c) \\
\end{array} 
\right. \]
\end{definition}

Note that we compute this array relative to a set of edges $E''$. These edges are slightly truncated versions of the edges in $E'$, and we will compute this set during the algorithm.

\begin{algorithm}[H]
\caption{$\mathsf{ClusterSensitivity}(G, E', T, C_\tau)$: Simulates MST sensitivity algorithm on $C_\tau$}
\label{alg:sens:pettie}
\begin{algorithmic}[1]
\State $E'' \gets \emptyset$
    \For{each non-tree edge $e = \{u, v\} \in E'$ (where $v$ is an ancestor of $u$)}
        \State Let $v'$ be the first node on a path from $v$ to $u$ in $T$.

        \State $\mc(\{v', v\} \gets \min \{ \mc(\{v', v\}), w(e)\}$

        \State Add an edge $\{u, v'\}$, with weight $w(e)$ to $E''$.
    \EndFor

    \For{each non-root cluster $c \in C_\tau$}, 
    \For{for each $i\in [D_T]$}
    
    determine the value $A_c[i]$
    
    \EndFor
\EndFor
    
\For {each non-root cluster $c \in C_\tau$ (i.e., edge $\{c,p(c)\}$ in $T_C$)}

    \State $\mathsf{minA}(c) \gets \min\{A_x[\lev(c)]: 
\mbox{$x$ is $c$ or a descendant of $c$}\}$

    \State $N_c \gets (\ell(p(c)), p(\ell(c)), \tau, \mathsf{minA}(c))$

    \State $\mc(\{\ell(c),p(\ell(c))\})$ $\gets \min\left\{\mc (\{\ell(c),p(\ell(c))\}), \mathsf{minA}(c)\right\}$
\EndFor

\State\Return $N' = \{N_c : c \in C_\tau\}$: a set of root-to-leaf notes, one for each cluster in $C_\tau$
\end{algorithmic}
\end{algorithm}

\begin{lemma}
\label{lem:pettie_is_correct}
    After executing \Cref{alg:sens:pettie} we have correctly computed:
    \begin{enumerate}
        \item[(i)] $\mc$ values for each edge in $T_{C_\tau}$; and
        \item[(ii)] for each pair of clusters $c_c, c_p \in C_\tau$ such that $c_c$ is a child of $c_p$, the minimum weight of an edge in $E'$ which covers the path from $\ell(c_p)$ to $p(\ell(c_c))$.
    \end{enumerate}
\end{lemma}
\begin{proof}
For \textbf{part~(i)}, it is enough to show that \Cref{alg:sens:pettie} correctly computes the \mc values for each edge in $T_C$.
Recall that the \mc value of a tree edge $e$ is the minimum weight of any non-tree edge which covers $e$.
Note that lines 1--5 of \Cref{alg:sens:pettie} split each edge into two parts: the topmost ``arc'' $\{v, v'\}$ of the edge $\{u, v\}$, and the rest of the edge ($\{u, v'\}$).
By line~4, $\mc(v, v')$ is at most the weight of the topmost arc, and so it remains to consider whether lines 6--11 correctly compute the \mc values for each tree edge relative to the edges in $E''$ after line 5.

As we assume as a precondition that the vertices of any non-tree edges are in an ancestor-descendant relation, we have, for each edge $\{u, p(u)\} \in T_C$:
\begin{align*}
 \mc & \left(\{u, p(u)\}\right) \\ 
     =& \min(\mc(u, p(u)), \min\{w(\{v,v'\}): 
       \mbox{$v$ is $u$ or a descendant of $u$, and $v'$ is an ancestor of $u$}\})\\
     =& \min(\mc(u, p(u)), \min\limits_{\mbox{$v$ is $u$ or a descendant of $u$}} 
     \min \{w(\{v,v'\}): \mbox{$v'$ is an ancestor of $u$}\})
\end{align*}

Applying the definition of $A_u[i]$, 
\begin{equation*}
       \mc \left(\{u,p(u)\}\right)=\min\bigl(\mc(u, p(u)), \min\limits_{\mbox{$v$ is either $u$ or a descendant of $u$}} A_v[\lev(u)]\bigr).
\end{equation*}

For \textbf{part~(ii)}, we argue that the minimum weight non-tree edge covering the path from $\ell(c_p)$ to $p(\ell(c_c))$ is exactly $\mathsf{minA}(c_c)$.
First, note that we have already shown that $\mathsf{minA}(c_c)$ is the minimum weight of a non-tree edge in $E''$ which covers the tree edge $\{\ell(c_c), \\ p(\ell(c_c)\}$.
We prove part~(ii) of the lemma by showing that the set of edges in $E''$ which cover the tree edge $\{\ell(c_c), p(\ell(c_c))\}$ is the same as the set of edges in $E''$ which cover the path from $\ell(c_p)$ to $p(\ell(c_c))$.

Fix some $e = \{u, v\} \in E'$, and suppose that $v$ is an ancestor of $u$.
Let $c_1 = c_\tau(u)$, $c_k = c_\tau(v)$, and let $c_2\dots c_{k-1}$ be the other clusters on a path from $c_1$ to $c_k$ in $T_{C_\tau}$.

Then, note that, since in \Cref{alg:sens:cluster_contractions} we maintain the invariant that all edges in $E'$ cover no edges in the clusters containing their endpoints, $u$ must be the root of $c_1$ and $p(\ell(c_{k-1})) = v$.
It it easy to see that $e$ covers the edges between $c_1$ and $c_2$, $c_2$ and $c_2$ \dots $c_{k-1}$ and $c_k$.
It also covers the path from $p(\ell(c_1))$ to $\ell(c_2)$, from $p(\ell(c_2))$ to $\ell(c_3)$\dots, and from $p(\ell(c_{k-2}))$ to $\ell(c_{k-1})$. 

Edges in $E''$ are the same as edges in $E'$, except that the topmost arc (in the example above, the edge between $c_{k-1}$ and $c_k$) have been removed.
It is easy to check that once this is done, iff an edge in $E''$ covers the edge between $c_i$ and $c_{i+1}$, it also covers the path from $p(\ell(c_i))$ to $\ell(c_{i+1})$, and so we are done.
\end{proof}

Next, we argue that we can implement \Cref{alg:sens:pettie} in $O(\log \dt)$ rounds of sublinear \MPC, with optimal global space.

\begin{lemma}
\label{lem:sens:computing_pettie}
    \Cref{alg:sens:pettie} can be implemented deterministically on an \MPC with $\lspace = O(n^\delta)$ and $\gspace = O(n + m)$ in $O(\log \dt)$ rounds.
\end{lemma}
\begin{proof}
We can compute the depth $\lev(c)$ in $T_{C_\tau}$ of each cluster $c \in C_\tau$, and we have already computed the DFS interval labeling of $\ell(c)$.
These steps can be implemented in $O(\log D_T)$ rounds.
Recall that, with the DFS interval labeling of two clusters $c$ and $c'$, we can decide whether $c$ and $c'$ are in an ancestor-descendant relationship.

For lines 1--6 of \Cref{alg:sens:pettie}, we can sort the edges in $T_C$ along with the non-tree edges in $E'$ by the topmost endpoint.
Let $\{u, v\}$ be a non-tree edge in $E'$ where $v$ is an ancestor of $u$.
By the invariant maintained in \Cref{alg:sens:cluster_contractions}, the child of $v$ on the path from $v$ to $u$ in $T$ is the root of a cluster.
Therefore, using prefix-sum, we can then identify which child $v'$ of $v$ in $T$ is on the path from $v$ to $u$ in $T$, for all such edges in parallel.
Lines~4 and~5 then follow easily.

For lines 7--10 of \Cref{alg:sens:pettie}, let us dedicate a set of consecutive machines $M_{c,i}$ for each non-root cluster $c \in C_\tau$ and each depth $i \in [D_T]$.
For each $c$ and $i \in [D_T]$, all the edges in $E''$ from a node in $c$ to an ancestor can be gathered onto $M_{c,i}$ using sorting in $O(1)$ rounds.
Then, for each $c$ and $i$, we can find the minimum weight of an edge on $M_{c, i}$.

For lines 11--15 of \Cref{alg:sens:pettie}, we can compute $\mathsf{minA}(c)$ for all clusters $c \in C$ using $O(\dt)$ parallel applications of the algorithm in \cite{dp_trees_mpc} for finding in parallel the minimum label in all subtrees in a tree.
This requires an additional $O(\dt)$ global space for each node in $T_C$, but since there are $O(n/\dt)$ nodes in $T_C$, this is a total global space of $O(m+n)$, as required.

Note that the total space used in our \MPC implementation is $O(nD_T)$ since we are storing an array of size $D_T$ for each non-root vertex.
\end{proof}

\subsection{Unwinding Cluster Contraction and Computing Sensitivity}\label{sec:unwind}

For all edges in $T_{C_\tau}$, we now have computed the $\mc$ (and thus \sens) values.
It remains to compute the $\mc$ values of the vertices within the vertices of $T_{C_\tau}$, which correspond to the clusters of $C_\tau$.

As part of our \MST sensitivity algorithm in the previous section, for each cluster $c \in C_\tau$, for each child $c_1, c_2\dots c_k$ of $c$ in $T_{C_\tau}$, we computed a root-to-leaf note for the paths from $\ell(c)$ to $p(\ell(c_1)), p(\ell(c_2)) \dots, p(\ell(c_k))$.
In this section we combine these with the root-to-leaf notes which were set aside during the cluster contraction part of the algorithm (\Cref{alg:sens:cluster_contractions}), and ensure that the \mc value for each intra-cluster edge is at most the minimum weight of a root-to-leaf note which covers that edge, when we undo the cluster contraction steps.

Our algorithm is as follows.

\begin{algorithm}
\caption{$\mathsf{UndoContractions}(G, T, N, C_\tau)$: Undo the cluster contractions and compute $\mc$ values for all edges}
\label{alg:sens:undo_contractions}
\begin{algorithmic}[1]
    \For{$i \in (\tau = O(\log \dt), \dots, 2, 1)$}
        \For{Each cluster $c$ formed in iteration $i$ of \Cref{alg:sens:cluster_contractions}, in parallel}
            \State Let $c_p$ be the senior sub-cluster of $c$; let $c_1\dots c_k$ be the junior sub-clusters of $c$.

            \For{Each root-to-leaf note $r = (\ell(c), v, i, w)$, in parallel}
                \State $c_i \gets$ The subcluster of $c$ containing $v$.

                \State $\mc(\{\ell(c_i), p(\ell(c_i))\}) \gets \min\{\mc(\{\ell(c_i), p(\ell(c_i))\}), w\}$.

                \State Remove root-to-leaf note $r$ from $N$.

                \State Add root-to-leaf note $(\ell(c), p(\ell(c_i)), I_{c_p}, w)$ to $N$, where $I_{c_p}$ is the iteration in which $c_i$ was formed.
                
                \State Add root-to-leaf note $(\ell(c_i), v, I_{c_i}, w)$ to $N$, where $I_{c_i}$ is the iteration in which $c_i$ was formed.
            \EndFor
        \EndFor

    \State Remove duplicate root-to-leaf notes (those whose first three values are the same), keeping the one with minimum weight.
    \EndFor
\end{algorithmic}
\end{algorithm}

We first argue that the \mc values which are computed during \Cref{alg:sens:undo_contractions} are correct:

\begin{lemma}
\label{lem:sens:correct_mc_values}
    The \mc values computed at the end of \Cref{alg:sens:undo_contractions} are correct: i.e.,~for all $v \in V$, $\mc(v)$ is the cost of the minimum edge which non-tree edge which covers (\Cref{def:covering}) $v$.
\end{lemma}
\begin{proof}
    Firstly, note that for all edges in $T_{C_\tau}$ (those edges which aren't ``inside'' clusters), \mc values have been computed correctly (\Cref{lem:pettie_is_correct}~(i)).

    We now argue that for all edges in $T \setminus T_{C_\tau}$, \mc values have been correctly computed.
    To do this, we look from the other perspective and show that each edge in $E \setminus T$ is at least as large as the \mc value for each edge in $T \setminus T_{C_\tau}$ which it covers.

    Let $e = \{u, v\} \in E \setminus T$ be a non-tree edge and let $P$ be the path in $T$ between $u$ and $v$.
    Recall that the endpoints of all non-tree edges are between ancestors and descendants.
    Note that $P$ comprises edges in $T_{C_\tau}$ (which we have already accounted for) and edges inside clusters in $C_\tau$.
    Let $c \in C_{\tau}$ be a cluster.
    There are two cases:
    \begin{itemize}
        \item Suppose $u$ is an ancestor of $\ell(c)$ and $v$ is in a descendant cluster of $c$.
        Then the section of the path which passes through $c$ must go from the root of $c$ to a leaf of $c$.
        By \Cref{lem:pettie_is_correct}~(ii), the minimum weight of an edge in $E'$ which covers each root-to-leaf path in $c$ has been correctly computed, and $e$ must weigh at least this value.
        
        \item Suppose that $u$ is inside $c$ and $P$ covers at least one edge in $c$.
        Then, during the cluster contraction process (\Cref{alg:sens:cluster_contractions}), (possibly several) root-to-leaf notes were created to account for this end-segment of $P$ which is inside a cluster.
        The only place in our subroutines where root-to-leaf notes are deleted is when another note covering the same path but with a lower weight is also present; therefore $e$ must weight at least as much as the root-to-leaf note covering this segment of $P$.
    \end{itemize}

    Therefore root-to-leaf notes of at most weight $w(e)$ are created for every edge in $P \setminus T_{C_\tau}$ (since edges in $T_{C_\tau}$ are accounted for).

    It remains to show that, for each edge $e$ on the path covered by a root-to-leaf note $(\cdot, \cdot, \cdot, w)$, that $\mc(e) \leq w$.
    This is ensured by lines 6--9 of \Cref{alg:sens:undo_contractions}.
    Root-to-leaf notes are split into three parts: a root-to-leaf note in the senior sub-cluster; a root-to-leaf note in a junior sub-cluster; and the edge between the two has its \mc value updated, if applicable.
    Eventually all edges on the path of a root-to-leaf notes will have their \mc value updated to the weight of that root-to-leaf note (if applicable), and will finish with at most that weight.
\end{proof}

Finally, we show that \Cref{alg:sens:undo_contractions} can be implemented in sublinear \MPC in $O(\log \dt)$ rounds and with optimal global memory:

\begin{lemma}
    \label{lem:sens:computing_mc_values}
    \Cref{alg:sens:undo_contractions} can be implemented deterministically on an \MPC with $\lspace = O(n^\delta)$ and $\gspace = O(n+m)$ in $O(\log \dt)$ rounds.
\end{lemma}
\begin{proof}
    We start by claiming that, at all times, $N$ contains at most $O(n)$ root-to-leaf notes.

    \begin{claim}
    \label{clm:lem:rtl_notes_always_linear}
        Throughout the execution of \Cref{alg:sens:undo_contractions}, $\size{N} = O(n)$.
    \end{claim}

    As a simple consequence of \Cref{obs:hierarchical_clustering_size}, we can, for each cluster, maintain a record of its subclusters and which iteration they were formed in.
    This information can be stored along with each cluster in $C_i$.
    When evaluating which \mc values to update and how and which root-to-leaf notes to create, machines have this information locally.
    
    We can set aside a set of machines dedicated to keeping track of and updating \mc values for all tree edges.
    Since there are at any time $O(n)$ root-to-leaf notes being dealt with in the for-loop from lines 4--10 (\Cref{clm:lem:rtl_notes_always_linear}), there are at most $O(n)$ updates to these values per loop (line 6).
    These machines can use sorting and prefix-sum to update their stored \mc values in $O(1)$ rounds and $O(n)$ memory.

    \Cref{clm:lem:rtl_notes_always_linear} implies that when deduplicating the root-to-leaf notes (which can be done in $O(1)$ rounds using sorting and prefix-sum), the global memory bounds are not exceeded.

    It remains to prove \Cref{clm:lem:rtl_notes_always_linear}:

    \begin{proof}[Proof~of~\Cref{clm:lem:rtl_notes_always_linear}]
        It suffices to bound the total possible number of non-duplicate root-to-leaf notes at any point, since when we potentially create additional duplicate notes (lines~8 and 9), this is only a constant-factor increase before we remove duplicates again in line 12.

        We will start by bounding the total number of root-to-leaf notes where the ``leaf'' has at least one child in $T$.
        Fix an iteration $i$ and a corresponding cluster decomposition $C_i$, and consider a root-to-leaf note $(r, l, i, w)$.
        All children of $l$ in $T$ (at least one exists by assumption) must be leaders of a cluster in $C_i$.
        Further, note that for each choice of $i$ and $l$ there is a unique $r$ (the root of the cluster containing $l$ in iteration $i$).
        Therefore, we can charge each choice of $i$ and $l$ (where $l$ has a child in $T$) to a vertex in $T_{C_i}$.
        Using \Cref{obs:hierarchical_clustering_size}, there at most $O(n)$ clusters over all levels of the hierarchical clustering.

        It remains to bound the number of root-to-leaf notes where the leaf is also a leaf of $T$.
        Consider one such leaf $l$.
        We claim, that, at any time, that there is at most one root-to-leaf note with $l$ as a leaf.
        This follows from two facts:
        \begin{itemize}
            \item Only one root-to-leaf note containing $l$ is created during \Cref{alg:sens:cluster_contractions}. Once a root-to-leaf note containing $l$ is created (and the edge in question is truncated), case~4 or 5 of \Cref{def:sensitivity_contraction_process} cannot apply to $l$ anymore ($l$ will always be in a junior sub-cluster but will never be the root of it).

            \item When a root-to-leaf note is processed in \Cref{alg:sens:undo_contractions}, exactly one root-to-leaf note containing $l$ is created.
        \end{itemize}
        This completes the proof of the claim. 
    \end{proof}
\end{proof}


Finally, we conclude the section as a whole with the proof of \Cref{thm:mst_sensitivity}:
\begin{proof}[Proof~of~\Cref{thm:mst_sensitivity}]
    By Observations \ref{obs:sensitivity_of_non_tree_edges} and \ref{obs:sensitivity_is_min_covering}, it suffices to compute the values of \mc for each tree-edge.
    It is easy to see that given the \mc values for all non-tree edges, we can compute the \sens values for all non-tree edges in $O(1)$ rounds.

    The correctness of \Cref{alg:sens} follows from \Cref{lem:pettie_is_correct,lem:sens:correct_mc_values}.
    Running time and memory constraints of \Cref{alg:sens} follow from \Cref{lem:sens:computing_cluster_contractions,lem:sens:computing_pettie,lem:sens:computing_mc_values}.

    All steps are deterministic except for the computation of the interval labeling (\Cref{lem:dfs_interval_labeling}), which means that the algorithm as a whole succeeds with high probability.
\end{proof}


\section{Lower Bound for \MST Verification}

In this section we give a simple lower bound result which shows that, conditioned on a popular 1-vs-2-cycle conjecture, our algorithm for \MST verification is optimal.

First, we recall the 1-vs-2-cycle conjecture.

\begin{conjecture}[\textbf{1-vs-2-cycle Conjecture}]
\label{1-vs-2-cycle Conjecture}
Let $\delta < 1$ be an arbitrary positive constant. The problem of distinguishing one cycle of length $n$ from two cycles of length $n/2$ requires $\Omega(\log n)$ rounds on an \MPC with local memory $\lspace = O(n^\delta)$.
\end{conjecture}

\begin{theorem}[\textbf{\MST Verification Lower Bound}]
\label{thm:verification_lb}
Let $\delta < 1$ be an arbitrary positive constant. Let $G = (V, E)$ be an edge-weighted graph. 
%
Let $T \subseteq E$ be a set of $n-1$ edges of $G$ and let $\dt$ be the diameter of $T$.

Deciding whether $T$ is a minimum spanning forest of $G$ requires $\Omega(\log \dt)$ rounds%
, on an \MPC with with local memory $\lspace = O(n^\delta)$, unless the 1-vs-2-cycle conjecture (\Cref{1-vs-2-cycle Conjecture}) is false.
\end{theorem}
\begin{proof}

    Let $G$ be a complete graph on some vertex set $V$, and suppose all edges in $G$ have weight $1$.
    Let $T \subset G$ be a spanning subgraph of $G$ with $n-1$ edges.
    Note that $D = \dmst = 1$, but $\dt$ might be much larger.

    Since all graphs with $n-1$ edges have the same weight, the problem of determining whether $T$ is a minimum spanning tree of $G$ is reduced to determining whether $T$ is connected.
    Suppose there is a $o(\log \dt)$ round algorithm for deciding whether $T$ is an \MST of $G$.
    This would imply that we can decide whether $T$ is connected in $O(1)$ rounds.
    However, by \cite{CoyC23}, this requires $\Omega(\log \dt)$ rounds, conditioned on the 1-vs-2-cycle conjecture.
\end{proof}

\bibliographystyle{alpha}
\bibliography{references}

\appendix

\section{Conditional MST Lower Bound}
\label{app:mst_lower_bound}

For the sake of completeness of our arguments, we sketch a simple proof of a claim that, conditioned on the 1-vs-2-cycles conjecture, there is no $o(\log n)$-rounds \MPC algorithm with $\lspace = O(n^{\delta})$ local memory and polynomial global memory $\gspace = \poly(n)$ for the minimum spanning tree problem.

\begin{claim}
\label{clm:mst_hardness}
Conditioned on the 1-vs-2 cycles \Cref{1-vs-2-cycle Conjecture}, there is no $o(\log \dmst)$-rounds \MPC algorithm for finding a minimum spanning tree, on sublinear \MPC with polynomial global memory.
\end{claim}
\begin{proof}
Let $G = (V,E)$ be an input graph with the promise that $G$ is either an $n$-vertex cycle or it consists of two disjoint $\frac{n}{2}$-vertex cycles. Define a new weighted graph $G^*$ with vertex set $V^* = V \cup \{v^*\}$ with $v^* \notin V$ and edge set $E^* = E \cup \{\{v^*,u\}: u \in V\}$, and with edge weight $w: E^* \rightarrow \mathcal{R}$ defined as follows:
\begin{align*}
    w(e) & =
    \begin{cases}
        1 & \text{ if } e \in E                      \\
        2 & \text{ if } e \in \{\{v^*,u\}: u \in V\} \\
    \end{cases}
\end{align*}
Notice that $G^*$ is connected, has $n+1$ vertices, $2n$ edges, and diameter $2$. The diameter of a minimum spanning tree is $\Theta(n)$ (whether $G$ is one or two cycles).

Next, observe that if $G$ is an $n$-vertex cycle then a minimum spanning tree algorithm, when run on $G^*$, returns a spanning tree of weight $n+1$, but if $G$ consists of two disjoint $\frac{n}{2}$-vertex cycles, then we get a spanning tree of weight $n+2$.

This implies that if there is a minimum spanning tree algorithm which runs in $o(\log n)$ rounds (even for constant diameter graphs) on sublinear \MPC with polynomial global memory then the 1-vs-2-cycles conjecture fails.

This does not preclude the existence of an $O(\log \dmst)$ round algorithm for the problem (in the proof, $O(\log \dmst) = \Omega(\log n)$).
In fact, such an algorithm exists with larger global memory \cite{ASSWZ18,CoyC23}.
The proof does rule out the existence of an $o(\log \dmst)$ algorithm however, since $o(\log \dmst) = o(\log n)$ for all values of \dmst. 
\end{proof}

\end{document}